\pdfoutput=1
\RequirePackage[l2tabu, orthodox]{nag}

\documentclass[paper=a4, fontsize=11pt,final]{scrartcl}

\usepackage[utf8]{inputenc}
\usepackage[T1]{fontenc}
\usepackage{textcomp}
\usepackage{lmodern}
\usepackage[english]{babel}
\usepackage{amsthm}
\theoremstyle{plain}
\newtheorem{theorem}{Theorem}[section]
\newtheorem{proposition}[theorem]{Proposition}

\newcommand*{\altqed}{\hfill\small\ensuremath{\triangleleft}}

\newtheorem{lemma}[theorem]{Lemma}
\newtheorem{corollary}[theorem]{Corollary}
\theoremstyle{definition}
\newtheorem{definitionx}[theorem]{Definition}
\newtheorem{example}[theorem]{Example}

\newenvironment{definition}
  {\pushQED{\qed}\definitionx}
  {\popQED\enddefinitionx}

\usepackage[final,backref=none,pagebackref=false,hyperindex=true,hyperfootnotes=true,bookmarks=true,pdfpagelabels=true,
bookmarksnumbered=true,
bookmarksopen=true,
bookmarksopenlevel=1,
]{hyperref}

\usepackage{bookmark}

\usepackage{authblk}

\newcommand*{\vv}[1]{\vec{\mkern1mu#1}}

\usepackage{bookmark}

\usepackage{pdfcolmk}

\usepackage[heavycircles]{stmaryrd}
\usepackage{xcolor}
\usepackage{enumitem}

\usepackage{braket}

\usepackage[babel,autostyle=true,strict]{csquotes}
\MakeOuterQuote{"}

\usepackage[
   centertags, 
   sumlimits,  
   intlimits,  
   namelimits, 
   fleqn,      
]{amsmath} 
\usepackage{amsfonts,amssymb}

\usepackage[fixamsmath,disallowspaces]{mathtools}

\usepackage{stackengine}

\usepackage{fixmath}

\usepackage[final,activate={true,nocompatibility},tracking=true,kerning=true,spacing=true,factor=1100,stretch=10,shrink=10]{microtype}

\usepackage[xspace]{ellipsis}

\usepackage{relsize}
\usepackage{scalerel}

\usepackage{booktabs}
\usepackage{array}
\newcolumntype{L}{>{$}l<{$}}
\newcolumntype{C}{>{$}c<{$}}

\usepackage{tikz}
\usetikzlibrary{arrows}
\usetikzlibrary{fit}
\usetikzlibrary{positioning}
\usetikzlibrary{patterns}
\usetikzlibrary{shapes,backgrounds}
\usetikzlibrary{decorations.pathmorphing}
\usetikzlibrary{decorations.pathreplacing}

\usepackage{amsthm}
\makeatletter
\newtheorem*{rep@theorem}{\rep@title}
\newcommand{\newreptheorem}[2]{
\newenvironment{rep#1}[1]{
 \def\rep@title{#2 \ref{##1}}
 \begin{rep@theorem}}
 {\end{rep@theorem}}}
\makeatother
\newreptheorem{proposition}{Proposition}
\newreptheorem{lemma}{Lemma}
\newreptheorem{theorem}{Theorem}

\usepackage[backend=biber,style=numeric,sortcites,url=true,doi=false,maxbibnames=99,maxcitenames=2,citetracker=true]{biblatex}

\usepackage{upgreek}

\DeclareFieldFormat[article]{title}{\mkbibemph{#1}}
\DeclareFieldFormat[article]{journal}{#1}
\DeclareFieldFormat[article]{journaltitle}{#1}
\DeclareFieldFormat[article]{volume}{\textbf{#1}}
\DeclareFieldFormat[article]{number}{no.~#1}

\DeclareFieldFormat[incollection]{title}{\mkbibemph{#1}}
\DeclareFieldFormat[incollection]{booktitle}{#1}
\DeclareFieldFormat[incollection]{volume}{\textbf{#1}}
\DeclareFieldFormat[incollection]{number}{(#1)}

\DeclareFieldFormat[inproceedings]{title}{\mkbibemph{#1}}
\DeclareFieldFormat[inproceedings]{booktitle}{#1}
\DeclareFieldFormat[inproceedings]{volume}{\textbf{#1}}
\DeclareFieldFormat[inproceedings]{number}{(#1)}

\DeclareFieldFormat[unpublished]{title}{\mkbibemph{#1}}

\DeclareFieldFormat[thesis]{title}{\mkbibemph{#1}}

\renewbibmacro*{volume+number+eid}{
  \printfield{volume}\space
  \printtext[parens]{\usebibmacro{date}}
  \setunit{\addcomma\space}
  \printfield{number}
  \setunit{\addcomma\space}
  \printfield{eid}
  }

\renewbibmacro*{issue+date}{
  \iffieldundef{issue}
      {}
      {\printfield{issue}}
  \newunit}

\renewbibmacro{in:}{}

\usepackage[linesnumbered]{algorithm2e}
\let\oldnl\nl
\newcommand{\nonl}{\renewcommand{\nl}{\let\nl\oldnl}}

\newcommand{\citet}[2][]{\citeauthor{#2}~\cite[#1]{#2}}

\newcommand{\ie}{i.e.\@\xspace}
\newcommand{\eg}{e.g.\@\xspace}
\newcommand{\wrt}{w.\,r.\,t.\@\xspace}
\newcommand{\wloss}{w.l.o.g.\@\xspace}
\newcommand{\Wloss}{W.l.o.g.\@\xspace}

\DeclareMathOperator*{\bigovee}{\scalerel*{\ovee}{\sum}}

\newcommand{\pow}[1]{\ensuremath{\mathfrak{P}}(#1)}
\newcommand{\size}[1]{{\ensuremath{\vert\nobreak#1\nobreak\vert}}}

\newcommand{\imp}{\rightarrow}

\newcommand{\arity}[1]{\mathsf{arity}(#1)}

\providecommand{\dfn}{\vcentcolon=}

\providecommand{\ddfn}{\vcentcolon\vcentcolon=}
\newcommand{\N}{\ensuremath{\mathbb{N}}\xspace}
\newcommand{\E}{\ensuremath\mathsf{E}}
\newcommand{\PS}{\ensuremath\mathcal{PS}}

\newcommand{\MC}{\ensuremath{\mathsf{MC}}}
\newcommand{\VAL}{\ensuremath{\mathsf{VAL}}}
\newcommand{\SAT}{\ensuremath{\mathsf{SAT}}}

\newcommand{\negg}{{\sim}}
\newcommand{\depop}{{=\!\!(\cdot,\cdot)}}
\newcommand{\dep}[1]{{=\!\!(#1)}}

\newcommand{\Var}{\mathsf{Var}}
\newcommand{\Prop}{\mathsf{Prop}}
\newcommand{\Fr}{\mathsf{Fr}}
\DeclareMathOperator{\dom}{\mathsf{dom}}

\newcommand{\hook}{\hookrightarrow}

\newcommand{\md}{\mathsf{md}}
\newcommand{\qr}{\mathsf{qr}}
\newcommand{\wi}{\mathsf{w}}

\newcommand{\logic}[1]{\ensuremath{\mathsf{#1}}\xspace}
\newcommand{\compclass}[1]{\ensuremath{\mathrm{#1}}\xspace}

\newcommand{\ML}{\logic{ML}}
\newcommand{\PTL}{\logic{PTL}}
\newcommand{\FO}{\logic{FO}}
\newcommand{\FOC}{\logic{FOC}}
\newcommand{\FON}{\ensuremath{\logic{FO}(\negg)}}
\newcommand{\FONN}{\ensuremath{\logic{FO}^2(\negg)}}
\newcommand{\FODK}[2]{\ensuremath{\logic{FO}^{#1}_{#2}(\calD)}}
\newcommand{\FONK}[2]{\ensuremath{\logic{FO}^{#1}_{#2}(\negg)}}
\newcommand{\FONKD}[2]{\ensuremath{\logic{FO}^{#1}_{#2}(\negg,\calD)}}
\newcommand{\FOND}{\ensuremath{\logic{FO}(\negg,\calD)}}
\newcommand{\FOD}{\ensuremath{\logic{FO}(\calD)}}
\newcommand{\SO}{\logic{SO}}
\newcommand{\MTL}{\logic{MTL}}
\newcommand{\D}{\logic{D}}
\newcommand{\TL}{\logic{TL}}

\newcommand{\AEXPPOLY}{\ensuremath{\compclass{ATIME}\text{-}\compclass{ALT}(\mathrm{exp},\mathrm{poly})}}
\newcommand{\TOWERPOLY}{\ensuremath{\compclass{TOWER}(\mathrm{poly})}}
\newcommand{\KAX}{\ensuremath{\compclass{ATIME}\text{-}\compclass{ALT}(\mathrm{exp}_k,\mathrm{poly})}}
\newcommand{\KAXX}[1]{\ensuremath{\compclass{ATIME}\text{-}\compclass{ALT}(\mathrm{exp}_{#1},\mathrm{poly})}}
\newcommand{\NEXPTIME}{\compclass{NEXPTIME}}

\newcommand{\PSPACE}{\compclass{PSPACE}}

\newcommand{\PTIME}{\compclass{PTIME}}

\newcommand{\calA}{\ensuremath{\mathcal{A}}}
\newcommand{\calB}{\ensuremath{\mathcal{B}}}
\newcommand{\calC}{\ensuremath{\mathcal{C}}}
\newcommand{\calD}{\ensuremath{\mathcal{D}}}

\newcommand{\calJ}{\ensuremath{\mathcal{J}}}
\newcommand{\calL}{\ensuremath{\mathcal{L}}}
\newcommand{\calK}{\ensuremath{\mathcal{K}}}

\newcommand{\calR}{\ensuremath{\mathcal{R}}}

\newcommand{\calX}{\ensuremath{\mathcal{X}}}

\newcommand{\bigO}[1]{\ensuremath{{\mathcal{O}(#1)}}}

\newcommand{\leqlogm}{\ensuremath{\leq^\mathrm{log}_\mathrm{m}}}

\newcommand{\rest}[2]{#1{\upharpoonright}#2}
\newcommand{\vval}[2]{\vv{#1}\langle{}#2\rangle{}}

\newcommand{\stx}{\mathsf{st}_x}
\newcommand{\sty}{\mathsf{st}_y}

\hypersetup{
    final,
    colorlinks,
    linkcolor={red!80!black},
    citecolor={blue!50!black},
    urlcolor={blue!80!black},
    pdftitle = {Canonical Models and the Complexity of Modal Team Logic},
    pdfauthor = {Martin Lück}
}

\author{\large Martin Lück\\\large Leibniz Universität Hannover, Germany\\ \large \texttt{lueck@thi.uni-hannover.de}}
\date{\vspace{-5ex}}

\title{On the Complexity of Team Logic and its Two-Variable Fragment}

\bibliography{main}

\begin{document}
\maketitle

\begin{abstract}
\textbf{Abstract.} We study the logic FO(${\sim}$), the extension of first-order logic with team semantics by unrestricted Boolean negation.
It was recently shown axiomatizable, but otherwise has not yet received much attention in questions of computational complexity.

In this paper, we consider its two-variable fragment FO$^2$(${\sim}$) and prove that its satisfiability problem is decidable, and in fact complete for the recently introduced non-elementary class TOWER(poly).

Moreover, we classify the complexity of model checking of FO(${\sim}$) with respect to the number of variables and the quantifier rank, and prove a dichotomy between PSPACE- and ATIME-ALT(exp, poly)-completeness.

To achieve the lower bounds, we propose a translation from modal team logic MTL to FO$^2$(${\sim}$) that extends the well-known standard translation from modal logic ML to FO$^2$.
For the upper bounds, we translate to a fragment of second-order logic.
\end{abstract}


\medskip

\noindent\textit{2012 ACM Subject Classification:} Theory of computation $\to$ Complexity theory and logic; Logic;

\section{Introduction}

In the last decades, the work of logicians has unearthed a plethora of decidable fragments of first-order logic $\FO$.
Many of these cases are restricted quantifier prefixes, such as the BSR-fragment which contains only $\exists^*\forall^*$-sentences \cite{Ramsey1987}.
Others include the guarded fragment $\mathsf{GF}$ \cite{guarded}, the recently introduced separated fragment $\mathsf{SF}$ \cite{separated1,separated2}, or the two-variable fragment $\FO^2$ \cite{Mortimer75,scott1962decision,fo2}.
See also \citet{BorgerGG1997} for a comprehensive classification of decidable and undecidable cases.

The above fragments all have been subject to intensive research with the purpose of further pushing the boundary of decidability.
One example is the extension of $\FO^2$ with counting quantifiers, $\FOC^2$.
For both logics, the satisfiability problem is $\NEXPTIME$-complete \cite{foc2,PrattHartmann10}.

\smallskip

Another actively studied extension of first-order logic is \emph{team semantics}, first studied by \citet{hodges_compositional_1997}, which at its core is the simultaneous evaluation of formulas on whole \emph{sets} of assignments, called \emph{teams}.
Based on team semantics, \citet{vaananen_dependence_2007} introduced dependence logic $\D$, which extends $\FO$ by atomic formulas $\dep{x_1,\ldots,x_n,y}$ called \emph{dependence atoms}.
Intuitively, these atoms state that the value of the variable $y$ in the team functionally depends on the values of $x_1,\ldots,x_n$.

On the level of expressive power, $\D$ coincides with existential second-order logic $\SO(\exists)$, and consequently is undecidable as well \cite{KontinenV09,vaananen_dependence_2007}.
Nonetheless, its two-variable fragment $\D^2$, as well as its $\exists^*\forall^*$-fragment, were recently proven by \citeauthor{d2nexp} to have a decidable satisfiability problem.
In particular, satisfiability of $\D^2$ is $\NEXPTIME$-complete due to a satisfiability-preserving translation to $\FOC^2$ \cite{d2nexp}.
Curiously, the corresponding \emph{validity} problem of $\D^2$ is undecidable~\cite{il2}, which is possible since dependence logic lacks a proper negation operator.
\citet{vaananen_dependence_2007} also considered \emph{team logic} $\TL$, the extension of $\D$ by a Boolean negation $\negg$.
This logic is equivalent to full second-order logic $\SO$~\cite{tl-to-so}, and the satisfiability problem is undecidable already for $\TL^2$.

\smallskip

We study the extension $\FO(\negg)$ of $\FO$ with team semantics and negation $\negg$, but without dependence atoms.
It was introduced by \citet{abramsky_strongly_2016}, and is incomparable to $\D$ in terms of expressive power.
Regarding its complexity, however, $\FO(\negg)$ is much weaker than $\D$.
In particular, its validity problem is---as for ordinary $\FO$---complete for $\Upsigma^0_1$ \cite{axiom-fo}, whereas it is non-arithmetical for $\D$ and $\TL$~\cite{vaananen_dependence_2007}.

In this paper, we show that its two-variable fragment $\FONN$ is decidable.
More precisely, its satisfiability and validity problem of $\FONN$ are complete for the recently introduced non-elementary complexity class $\TOWERPOLY$~\cite{mtl_report}.

Our approach for the satisfiability problem is to establish both a finite model property of $\FONN$ and an upper bound for model checking.
In fact, the complexity of the model checking problem has received much less attention in team semantics so far; recent progress focused on the propositional~\cite{hannula_complexity_2018} and modal~\cite{emdl,julian,yang_extensions_2014} variants of team logic and dependence logic.

In the first-order case, \citet{mc-games} proved that model checking for $\D$ is $\NEXPTIME$-complete.
Extending this result, we show that model checking is $\AEXPPOLY$-complete for $\FO(\negg,\calD)$, where $\calD$ is any set of $\FO$-definable dependencies.
Examples for such dependencies are the aforementioned dependence atom $\depop$, but also atoms such as independence $\perp$~\cite{indep} or inclusion $\subseteq$~\cite{galliani_inex}.

\begin{table*}\centering
\begin{tabular}{LCCl}
\toprule
\text{Logic} & \text{Satisfiability} & \text{Validity} & References\\
\midrule
\FO^2 & \NEXPTIME & \text{co-}\NEXPTIME & \cite{fo2}\\
\FO^2(\negg) & \TOWERPOLY & \TOWERPOLY & Theorem~\ref{thm:main-sat} \\
\D^2 & \NEXPTIME & \Upsigma^0_1\text{-hard} & \cite{il2}\\
\TL^2 & \Uppi^0_1\text{-hard} & \Upsigma^0_1\text{-hard} & \cite{il2}\\
\midrule
& \multicolumn{2}{c}{\text{Model Checking}}  & \\
\midrule
\FO_k,\FO^n & \multicolumn{2}{c}{$\in \PTIME$} & \cite{fmt}\\[.8mm]
\FO & \multicolumn{2}{c}{\PSPACE} & \cite{fmt}\\
\midrule
\FONKD{}{k}, \FONKD{n}{} & \multicolumn{2}{c}{\PSPACE} & Theorem~\ref{thm:mc-completeness}\\[.8mm]
\FONKD{}{} & \multicolumn{2}{c}{\AEXPPOLY} & Theorem~\ref{thm:mc-completeness} \\
\bottomrule
\end{tabular}

\medskip
\caption{Complexity of logics with team semantics. Completeness unless stated otherwise.\label{tab:table-sat}}
\end{table*}

\medskip

\textbf{Organization.}
In Section~\ref{sec:prelim}, we introduce first-order team logic $\FO(\negg,\calD)$ for arbitrary dependencies $\calD$, as well as the fragment $\FONKD{n}{k}$ of bounded width $n$ and quantifier rank $k$.
Moreover, we restate the definition of modal team logic $\MTL$, which is used for several lower bounds.
In Section~\ref{sec:upper}, we prove the upper bounds for model checking by a translation to second-order logic $\SO$, which admits model checking in $\AEXPPOLY$.
Afterwards, the upper bound for the satisfiability problem follows in Section~\ref{sec:upper-sat}.
Finally, the lower bounds are proved in Section~\ref{sec:lower}, where we propose a translation from $\MTL$ to $\FONN$ that extends the well-known \emph{standard translation} from $\ML$ to $\FO^2$.

For better readability, some proofs are moved to the appendix and marked with $(\star)$.

\section{Preliminaries}
\label{sec:prelim}

The domain of a function $f$ is $\dom f$.
For $f \colon X \to Y$ and $Z \subseteq X$, $\rest{f}{Z}$ is the restriction of $f$ to the domain $Z$.
The power set of $X$ is $\pow{X}$.

Given a logic $\calL$, the sets of all satisfiable resp.\ valid formulas of $\calL$ are denoted by $\SAT(\calL)$ and $\VAL(\calL)$, respectively.
Likewise, the model checking problem $\MC(\calL)$ is the set of all tuples $(A,\varphi)$ such that $\varphi$ is an $\calL$-formula and $A$ is a model of $\varphi$.
The notation $\size{\varphi}$ means the length of $\varphi$ in a suitable encoding.
$\omega$ is the cardinality of the natural numbers.

\smallskip

We assume the reader to be familiar with alternating Turing machines \cite{alternation}.
When stating that a problem is hard or complete for a complexity class $\calC$, we refer to $\leqlogm$-reductions.
We use the notation $\exp_k$ for the tetration operation, with $\exp_0(n) \dfn n$ and $\exp_{k+1}(n) \dfn 2^{\exp_{k}(n)}$.\label{p:exp}
We simply write $\exp(n)$ instead of $\exp_1(n)$.

\begin{definition}[\cite{mtl_report}]
For $k \geq 0$, $\KAX$ is the class of problems decided by an alternating Turing machine with at most $p(n)$ alternations and runtime at most $\exp_k(p(n))$, for a polynomial $p$.
\end{definition}
\begin{definition}[\cite{mtl_report}]
$\TOWERPOLY$ is the class of problems that are decided by a deterministic Turing machine in time $\exp_{p(n)}(1)$ for some polynomial $p$.
\end{definition}

The reader may verify that both $\KAX$ and $\TOWERPOLY$ are closed under polynomial time reductions (and in particular $\leqlogm$).

\subsection*{Modal Team Logic}

We fix a set $\Phi$ of propositional symbols.
\emph{Modal team logic} $\MTL$, introduced by \citet{julian}, extends classical modal logic $\ML$ via the following grammar, where $\varphi$ denotes an $\MTL$-formula, $\alpha$ an $\ML$-formula, and $p$ a propositional variable.
\begin{align*}
\varphi &\ddfn \negg \varphi \mid \varphi \land \varphi \mid \varphi \lor \varphi \mid \Box \varphi \mid \Diamond \varphi \mid \alpha\\
\alpha &\ddfn \neg \alpha \mid \alpha  \land \alpha  \mid \alpha  \lor \alpha \mid \Box\alpha \mid \Diamond\alpha \mid p
\end{align*}

For easier distinction, we usually call classical formulas $\alpha,\beta,\gamma,\ldots$ and reserve $\varphi,\psi,\vartheta,\ldots$ for general formulas.

The \emph{modal depth} $\md(\varphi)$ of formulas $\varphi$ is defined as usual:
\begin{alignat*}{3}
&\md(p) &&\dfn 0 &&\\
&\md(\negg \varphi) &&\dfn \md(\neg\varphi) &&\dfn \md(\varphi)\\
&\md(\varphi \land \psi) &&\dfn \md(\varphi \lor \psi) &&\dfn \mathrm{max}\{\md(\varphi),\md(\psi)\}\\
&\md(\Diamond \varphi) && \dfn \md(\Box\varphi) &&\dfn \md(\varphi) + 1
\end{alignat*}
The set of propositional variables occurring in $\varphi \in \MTL$ is denoted by $\Prop(\varphi)$.
$\MTL_k$ is the fragment of $\MTL$ with modal depth at most $k$.

\medskip

Let $X \subseteq \Phi$ be a finite set of propositions.
A \emph{Kripke structure} (over $X$) is a tuple $\calK = (W, R, V)$, where $W$ is a set of \emph{worlds}, $(W,R)$ is a directed graph, and $V \colon X \to \pow{W}$ is the \emph{valuation}.
If $w \in W$, then $(\calK,w)$ is called \emph{pointed Kripke structure}.
$\ML$ is evaluated on pointed structures in the classical Kripke semantics, whereas
$\MTL$ is evaluated on pairs $(\calK,T)$, \emph{Kripke structure with teams}, where $T \subseteq W$ is called \emph{team} (in $\calK$).

The team $RT \dfn \{ v \in W \mid \exists w \in T \colon Rwv \}$ is called \emph{image} of $T$, and we write $Rw$ instead of $R\{w\}$ for brevity.
A \emph{successor team} of $T$ is a team $S$ such that $S \subseteq RT$ and $T \subseteq R^{-1}S$, where $R^{-1} \dfn \{(v,w) \in W^2\mid Rwv\}$.
Intuitively, every $w\in T$ has at least one successor in $S$, and every $v \in S$ has at least one predecessor in $T$.

The semantics of formulas $\varphi \in \MTL$ is defined as
\begin{alignat*}{3}
&(\calK, T) \vDash \alpha && \Leftrightarrow\;\forall w \in T \colon (\calK, w) \vDash \alpha \; \text{ if } \alpha \in \ML\text{, and otherwise as}\\
&(\calK, T) \vDash \negg \psi &&\Leftrightarrow\;(\calK,T) \nvDash \psi\text{,}\\
&(\calK,T) \vDash \psi \land \theta&&\Leftrightarrow\;(\calK,T) \vDash \psi \text{ and }(\calK,T) \vDash \theta\text{,}\\
&(\calK,T) \vDash \psi \lor \theta&&\Leftrightarrow\;\exists S, U \subseteq T \text{ such that }T = S \cup U\text{, }(\calK,S) \vDash \psi\text{, and }(\calK,U)\vDash \theta\text{,}\\
&(\calK,T)\vDash \Diamond \psi&&\Leftrightarrow\;(\calK, S)\vDash \psi \text{ for some successor team }S\text{ of }T\text{,}\\
&(\calK,T)\vDash \Box\psi&&\Leftrightarrow\;(\calK,RT) \vDash \psi\text{.}
\end{alignat*}
A formula $\varphi \in \MTL$ is \emph{satisfiable} if $(\calK,T) \vDash \varphi$ for some Kripke structure $\calK$ over $X\supseteq \Prop(\varphi)$ and team $T$ in $\calK$.
Likewise, it is \emph{valid} if it is true in every such pair.

\smallskip

The modality-free fragment $\MTL_0$ syntactically coincides with \emph{propositional team logic} $\PTL$ \cite{gandalf,hannula_complexity_2018,YANG20171406}.
The usual interpretations of the latter, \ie, sets of Boolean assignments, can easily be represented as teams in Kripke structures.
For this reason, we identify $\PTL$ and $\MTL_0$ in this paper.

\smallskip

It is easy to check that the following lower bounds due to \citet{hannula_complexity_2018} are logspace reductions.

\begin{theorem}[\cite{hannula_complexity_2018}]
\label{thm:mc-ptl}
$\MC(\PTL)$ is $\PSPACE$-complete.
\end{theorem}

\begin{theorem}[\cite{hannula_complexity_2018}]
\label{thm:sat-ptl}
$\SAT(\PTL)$ and $\VAL(\PTL)$ are $\AEXPPOLY$-complete.
\end{theorem}

For each increment in modal depth, the complexity of the satisfiability problem increases by an exponential, reaching a non-elementary class in the unbounded case:

\begin{theorem}[\cite{mtl_report}]\label{thm:mtl-sat}
$\SAT(\MTL)$ and $\VAL(\MTL)$ are $\TOWERPOLY$-complete.
For every finite $k \geq 0$, $\SAT(\MTL_k)$ and $\VAL(\MTL_k)$ are $\KAXX{k+1}$-complete.
\end{theorem}
 
\subsection*{First-order Team Logic}

A vocabulary is a (possibly infinite) set $\tau$ of function $f$ and predicate symbols $P$, each with arity $\arity{f}$ and $\arity{P}$, respectively.
$\tau$ is called \emph{relational} if it contains no function symbols.
We explicitly state ${=}\in \tau$ if we admit equality.
Moreover, we require that a vocabulary always contains at least one predicate or $=$.

We fix a set $\Var = \{x_1,x_2,\ldots\}$ of first-order variables.
If $\vec{t}$ is a tuple of $\tau$-terms, we let $\Var(\vv{t})$ denote the set of variables occurring in $\vv{t}$.
$\tau$-structures are pairs $\calA = (A,\tau^\calA)$, where $\dom \calA \dfn A$ is the domain of $\calA$.
We sometimes briefly write $\calA$ instead of $\dom\calA$ if the meaning is clear.
If $s \colon X \to \calA$ and $\dom s \supseteq \Var(\vv{t})$, then $\vv{t}\langle s\rangle \in \calA$ is the evaluation of the terms $\vv{t}$ in $\calA$ under $s$.

\smallskip

A \emph{team} $T$ (in $\calA$) is a set of assignments $s\colon X \to \calA$, where $X$ is the \emph{domain} of $T$.
If $X \supseteq \Var(\vv{t}\,)$ and $\vv{t}$ is a tuple of terms,  $\vv{t}\langle  T \rangle \dfn \{ \vv{t}\langle  s\rangle \mid s \in T \}$.
If $T$ is a team with domain $X \supseteq Y$, then its \emph{restriction} to $Y$ is $\rest{T}{Y}\dfn\{ \rest{s}{Y} \mid s \in T\}$.
In slight abuse of notation, we sometimes identify a tuple $\vv{x}$ with its underlying set, \eg, write $\rest{T}{\vv{x}}$ instead of $\rest{T}{\{x_1,\ldots,x_{\size{\vv{x}}}\}}$, or $\vv{x} \subseteq \vv{y}$ instead of $\{x_1,\ldots,x_{\size{\vv{x}}}\} \subseteq \{y_1,\ldots,y_{\size{\vv{y}}}\}$.

\begin{proposition}[$\star$]\label{prop:rel}
Let $\calA$ be a structure, $\vv{t}$ a tuple of terms, $X \supseteq \Fr(\vv{t}\,)$.
For $i \in \{1,2\}$, let $T_i$ be a team in $\calA$ with domain $X_i \supseteq X$.
Then $\rest{T_1}{X} = \rest{T_2}{X}$ implies $\vv{t}\langle T_1 \rangle = \vv{t}\langle T_2 \rangle$.
Furthermore, for any tuple $\vv{x} \subseteq X$ of variables, $\vval{x}{T_1} = \vval{x}{T_2}$ iff $\rest{T_1}{\vv{x}} = \rest{T_2}{\vv{x}}$.
\end{proposition}

\begin{proposition}[$\star$]\label{prop:iso}
Let $\calA$ be a structure, $\vv{x}$ a tuple of variables, and $V \dfn \{ s \colon \vv{x} \to \calA \}$.
Then $\pow{V}$ is the set of all teams in $\calA$ with domain $\vv{x}$, and the mapping $S \mapsto \vv{x}\langle S\rangle$ is an isomorphism
between the Boolean lattices $(\pow{V},\subseteq)$ and $(\pow{\calA^{\size{\vv{x}}}},\subseteq)$.
\end{proposition}

If $s \colon X \to \calA$ and $x \in \Var$, then $s^x_a \colon X \cup \{x\} \to \calA$ maps $x$ to $a$ and $y \in X \setminus \{x\}$ to $s(y)$.
If $T$ is a team in $\calA$ with domain $X$, then $f \colon T \to \pow{\calA}\setminus \{\emptyset\}$ is called a \emph{supplementing function} of $T$.
It extends (or modifies) $T$ to the team $T^x_f \dfn \Set{ s^x_a | s \in T, a \in f(s)}$ with domain $X \cup \{x\}$.
For $f(s) = A$ is constant, we write $T^x_A$ for $T^x_f$.

\smallskip

As an alternative definition of supplementing functions, \citeauthor{galliani_general_2013} used \emph{$x$-variations} \cite{galliani_general_2013}, which are teams that "agree" on all variables but $x$.

\begin{proposition}[$\star$]\label{prop:def-of-supplement}
Let $T$ be a team with domain $X$ and $S$ a team with domain $X \cup \{x\}$ (with possibly $x \in X$), and let $X' \dfn X \setminus \{x\}$.
Then $\rest{S}{X'} = \rest{T}{X'}$ if and only if there is a supplementing function $f$ such that $S = T^x_f$.
\end{proposition}

We are now finally ready to define first-order team logic.
First, the set of ordinary $\tau\text{-}\FO$-formulas is given by the grammar
\[
\alpha \ddfn P t_1\ldots t_{\arity{P}} \mid \neg \alpha \mid \alpha \land \alpha \mid \alpha \lor \alpha \mid  \exists x\, \alpha \mid \forall x\, \alpha\text{,}
\]
where $P \in \tau$ is a predicate (possibly ${=}$), $x \in \Var$ and $t_1,t_2,\ldots$ are $\tau$-terms.
If $\vv{t} = (t_1,\ldots,t_n)$ and $\vv{u} = (u_1,\ldots,u_n)$ are tuples of terms, then we use the shorthand $\vv{t} = \vv{u}$ for $\bigwedge_{i=1}^n t_i = u_i$.

\begin{definition}[Dependencies]
Let $\tau_P = \{ P\}$, where $P$ is a predicate of arity $k$.
Then $\delta \in \tau_P\text{-}\FO$ is a \emph{$k$-ary dependency}, and $\mathbf{D}_\delta \dfn \Set{ \tau_P\text{-structure }\calA \mid \calA \vDash \delta}$.
\end{definition}

Let $\calD = \{ \delta_1,\delta_2,\ldots\}$ be a (possibly infinite) set of dependencies.
The logic $\tau$-$\FOND$ then extends $\tau$-$\FO$ by
\[
\varphi \ddfn \alpha \mid A_i\vv{t} \mid \negg \varphi \mid \varphi \land \varphi \mid \varphi \lor \varphi \mid \exists x\, \varphi \mid \forall x\, \varphi\text{,}
\]
where $\alpha \in \tau\text{-}\FO$, $\delta_i \in \calD$ is a $k$-ary dependency, $\vv{t}$ is a $k$-tuple of $\tau$-terms, and $x \in \Var$.
The $A_i\vv{t}$ are called \emph{generalized dependency atoms} or simply \emph{dependency atoms}.

In what follows, we omit $\tau$ if it is not relevant.
The $\negg$-free fragment of $\FOND$ is $\FOD$, and we write $\FON$ instead of $\FOND$ if $\calD = \emptyset$.

\smallskip

\begin{example}
The set $\mathsf{dep} \dfn \{ \delta_1, \delta_2, \ldots \}$ of dependencies is defined by
\[
\delta_n(R) \dfn \forall x_1 \cdots \forall x_{n-1} \forall y \forall z (Rx_1\cdots x_{n-1}y \land Rx_1\cdots x_{n-1}z \imp y = z)\text{.}
\]
$A_n\vv{t}$ is called $n$-ary \emph{dependence atom} and is also written $\dep{t_1,\ldots,t_n}$.
\end{example}

In this notation, Väänänen's dependence logic $\D$ is $\FO(\mathsf{dep})$, and team logic $\TL$ is $\FO(\negg,\mathsf{dep})$ \cite{vaananen_dependence_2007}.
Other first-order-definable dependencies include the independence atom~\cite{gradel_dependence_2013} as well as the inclusion and exclusion atom \cite{galliani_inex} (see also~\citet{abramsky_expressivity_2016}).

\medskip

If $\varphi$ is a formula, $\Fr(\varphi)$ and $\Var(\varphi)$ denote the set of free resp.\ of all variables in $\varphi$,
defined in the usual manner.
We let $\Fr(A_i\vv{t}\,) \dfn \Var(A_i\vv{t}\,) \dfn \Var(\vv{t})$.
The \emph{width} $\wi(\varphi)$ of $\varphi$ is $\size{\Var(\varphi)}$.
The \emph{quantifier rank} $\qr(\varphi)$ of $\varphi$ is
\begin{alignat*}{3}
&\qr(\varphi) &&\dfn 0 && \text{ if $\varphi$ is a first-order atom or a dependency atom,}\\
&\qr(\negg \varphi) &&\dfn \qr(\neg\varphi) &&\dfn \qr(\varphi)\text{,}\\
&\qr(\varphi \land \psi) &&\dfn \qr(\varphi \lor \psi) &&\dfn \mathrm{max}\{\qr(\varphi),\qr(\psi)\}\text{,}\\
&\qr(\exists x \,\varphi) && \dfn \qr(\forall x\, \varphi) &&\dfn \qr(\varphi) + 1\text{.}
\end{alignat*}

The fragment of $\FO$ of formulas with width at most $n \leq \omega$ and quantifier rank at most $k \leq \omega$ is denoted by $\FO^n_k$, with $\FONKD{n}{k}$, $\FONK{n}{k}$ and $\FODK{n}{k}$ defined analogously.
$\tau$-$\FOND$-formulas $\varphi$ are evaluated on pairs $(\calA,T)$, where $\calA$ is a $\tau$-structure and $T$ a team $T$ in $\calA$ with domain $X \supseteq \Fr(\varphi)$:
\begin{alignat*}{3}
&(\calA,T) \vDash \varphi && \Leftrightarrow\;\forall s \in T : (\calA,s) \vDash \varphi \text{ (in Tarski semantics) if }\varphi \in \FO\text{,}\\
&(\calA,T) \vDash A_i\vv{t} && \Leftrightarrow\;\calA \vDash \delta_i(\vval{t}{T})\text{,}\\
&(\calA, T) \vDash \negg \psi &&\Leftrightarrow\;(\calA,T) \nvDash \psi\text{,}\\
&(\calA,T) \vDash \psi \land \theta&&\Leftrightarrow\;(\calA,T) \vDash \psi \text{ and }(\calA,T) \vDash \theta\text{,}\\
&(\calA,T) \vDash \psi \lor \theta&&\Leftrightarrow\;\exists S, U \subseteq T \text{ such that }T = S \cup U\text{, }(\calA,S) \vDash \psi\text{, and }(\calA,U)\vDash \theta\text{,}\\
&(\calA,T)\vDash \exists x \, \varphi\;&&\Leftrightarrow\; (\calA,T^x_f) \vDash \varphi \text{ for some }f \colon T \to \pow{\size{\calA}}  \setminus \{\emptyset\}\text{,}\\
&(\calA,T)\vDash \forall x \, \varphi\;&&\Leftrightarrow\;(\calA,T^x_{\size{\calA}}) \vDash \varphi\text{.}
\end{alignat*}

A $\tau$-formula $\varphi$ is \emph{satisfiable} if there exists a $\tau$-structure $\calA$ and team $T$ with domain $X \supseteq \Fr(\varphi)$ in $\calA$ such that $(\calA,T) \vDash \varphi$.
Likewise, $\varphi$ is \emph{valid} if $(\calA,T) \vDash \varphi$ for all $\tau$-structures $\calA$ and teams $T$ as above.

\medskip

For $\alpha \in \FO$ and $\varphi \in \TL$, define $\alpha \hook \varphi \dfn \neg \alpha \lor (\alpha \land \varphi)$, and if $T$ is a team in $\calA$, let $T_\alpha \dfn \Set{ s \in T \mid (\calA,s) \vDash \alpha }$.
This operator was introduced by \citet{galliani_upwards_2015} and \citet{tl-to-so}.

\begin{proposition}
$(\calA,T) \vDash \alpha \hook \varphi$ if and only if $(\calA,T_\alpha) \vDash \varphi$.
\end{proposition}
\begin{proof}
Straightforward.
See also \citeauthor{galliani_upwards_2015} \cite[Lemma 16]{galliani_upwards_2015}.
\end{proof}

\begin{proposition}[Locality]\label{prop:free-vars}
Let $\varphi \in \FOND$ and $X \supseteq \Fr(\varphi)$.
If $T$ is a team in $\calA$ with domain $Y \supseteq X$, then $(\calA,T) \vDash \varphi$ if and only if $(\calA,\rest{T}{X})\vDash \varphi$.
\end{proposition}
\begin{proof}
Proven by induction on $\varphi$.
The base case of $\FO$-formulas and the inductive step for $\land$, $\lor$, $\exists$ and $\forall$ works similarly as \citeauthor{galliani_inex}'s proof for inclusion/exclusion logic \cite[Theorem 4.22]{galliani_inex}, to which the $\negg$-case can be added in the obvious manner.

It remains to consider the dependence atoms $A_i\vv{t}$.
As $X \supseteq \Fr(A_i\vv{t}) = \Var(\vv{t})$, clearly $\vv{t}\langle s \rangle = \vec{t}\langle \rest{s}{X} \rangle$ for any $s \in T$, and consequently, $\vv{t}\langle T \rangle = \vv{t}\langle \rest{T}{X} \rangle$.
Hence, $\calA \vDash \delta_i(\vval{t}{T})$ iff $\calA \vDash \delta_i(\vval{t}{\rest{T}{X}})$.
\end{proof}

\subsection*{Second-Order Logic}

Second-order logic $\tau$-$\SO$ (or simply $\SO$) extends $\tau$-$\FO$ by an infinite set of function and predicate variables disjoint from $\tau$ and corresponding quantifiers:
\begin{align*}
\hspace{-1.2em}\alpha \ddfn \;&t_1 = t_2 \mid Pt_1\ldots t_n \mid \neg \alpha \mid \alpha \land \alpha \mid \alpha \lor \alpha \mid \exists x\, \alpha\mid \forall x\, \alpha \mid \exists f \, \alpha \mid  \forall f \, \alpha \mid \exists P \, \alpha  \mid \forall P \, \alpha \text{.}
\end{align*}
In second-order logic, $\Var(\alpha)$ resp.\ $\Fr(\alpha)$ refers to the set of all (free) element, function and predicate variables in $\alpha$.
Formulas are evaluated on pairs $(\calA,\calJ)$, where $\calA$ is a structure and $\calJ$ is a \emph{second-order assignment}, \ie, it maps first-order variables $x\in \Fr(\alpha)$ to elements $\calJ(x) \in \calA$, function variables $f\in\Fr(\alpha)$ to functions $\calJ(f) \colon \calA^{\arity{f}} \to \calA$, and predicate variables $P \in \Fr(\alpha)$ to relations $\calJ(P) \subseteq \calA^{\arity{P}}$.
The notations $\calJ^x_a$, $\calJ^f_F$ and $\calJ^P_R$ are defined like $s^x_a$ in the first-order setting.
Based on this, the semantics of $\SO$ are defined in the usual manner.
Instead of $(\calA,\calJ) \vDash \alpha(X_1,\ldots,X_n)$ and $\calJ(X_1) = \calX_1, \ldots, \calJ(X_n) = \calX_n$, we also write $\calA \vDash \alpha(\calX_1,\ldots,\calX_n)$.

Second-order model checking, $\MC(\SO)$, is straightforwardly decidable by evaluating formulas in recursive top-down manner, and using non-deterministic guesses for the quantified elements, functions and relations, which may be of at most exponential size.

\begin{proposition}[$\star$]\label{prop:so-mc}
$\MC(\SO)$ is decidable in time $2^{n^\bigO{1}}$ and with $\size{\alpha}$ alternations, where $(\calA,\calJ,\alpha)$ is the input.
\end{proposition}

If the arity of quantified functions and relations is bounded a priori, then the interpretation of each quantified function and relation contains $\size{\calA}^{\bigO{1}}$ tuples.

\begin{corollary}
Let $\tau$-$n$-$\SO$ be the fragment of $\tau$-$\SO$ where all quantified functions and relations have arity at most $n$.
Then $\MC(\tau\text{-}n\text{-}\SO)$ is $\PSPACE$-complete.
\end{corollary}
 \section{From FO($\negg$) to SO: Upper bounds for model checking}

\label{sec:upper}

In this section, we present complexity upper bounds for the model checking problem of team logic.
On that account, we assume all first-order structures $\calA$ and teams $T$ to be finite.

Instead of deciding $\MC(\FOND)$ directly, we reduce to the corresponding problem of second-order logic, $\MC(\SO)$.
For this purpose, we use a result of \citet{vaananen_dependence_2007}, which states that, roughly speaking, $\TL$-formulas can be translated to $\SO$.
As a consequence, the upper bounds of $\SO$ carry over to $\TL$.

Väänänen's original translation \cite[Theorem 8.12]{vaananen_dependence_2007} from $\TL$ to $\SO$ assumes that all variables in a formula are quantified at most once.
However, since we also consider fragments $\FONKD{n}{k}$ of finite width $n$, re-quantification of bound variables may be necessary.
In what follows, we generalize the translation accordingly, and also extend it to include generalized dependence atoms.

\medskip

Suppose $\vv{x} = (x_1,\ldots,x_n)$ is a tuple of variables, $\varphi\in \FOND$ such that $\Fr(\varphi) \subseteq \vv{x}$, and $R$ is a $n$-ary predicate.
Then we recursively define the $\SO$-formula $\eta^{\vv{x}}_\varphi(R)$ as shown below.
In order to avoid repetitions in the tuple $\vv{x}$, we define the notation $\vv{x}{;}y$ as
$\vv{x}$ if $y = x_i$ for some $1 \leq i \leq n$, and as $(x_1,\ldots,x_n,y)$ if $y \notin \{x_1,\ldots,x_n\}$.

Next, we state the translation recursively.

\begin{itemize}
\item If $\varphi$ is first-order, then $\eta^{\vec{x}}_\varphi(R) \dfn \forall \vec{x} (R\vec{x} \imp \varphi)$.\smallskip
\item If $\varphi = A_i(\vv{t})$ and $\delta_i \in \calD$ is $k$-ary, then let $\vv{z} = (z_1,\ldots,z_k)$ be pairwise distinct variables disjoint from $\vv{x}$, and
 $\eta^{\vec{x}}_\varphi(R) \dfn \exists S\,\forall \vv{z} \, (S\vv{z} \leftrightarrow (\exists \vv{x}\, R\vv{x} \land \vv{t} = \vv{z}\,)) \land \delta_i(S)$.\smallskip
\item If $\varphi = \negg \psi$, then $\eta^{\vec{x}}_\varphi(R) \dfn \neg \eta^{\vec{x}}_\psi(R)$.\smallskip
\item If $\varphi = \psi \land \theta$, then
$\eta^{\vec{x}}_\varphi(R) \dfn \eta^{\vec{x}}_\psi(R) \land \eta^{\vec{x}}_\theta(R)$.\smallskip
\item If $\varphi = \psi \lor \theta$, then
$\begin{aligned}[t]
\eta^{\vec{x}}_\varphi(R) \dfn \exists S\, \exists U \,&\,  \forall \vec{x} (R\vec{x} \leftrightarrow (S\vec{x} \lor U\vec{x}))\land \eta^{\vec{x}}_\psi(S) \land \eta^{\vec{x}}_\theta(U)\text{.}
\end{aligned}$\smallskip

\item If $\varphi = \exists y \, \psi$
, then
$\begin{aligned}[t]
\eta^{\vec{x}}_\varphi(R) \dfn \exists S\, \forall \vec{x} ((\exists y R\vec{x}) \leftrightarrow (\exists y\, S\vec{x}{;}y))
\land \eta_\psi^{\vec{x}{;}y}(S)\text{.}
\end{aligned}$\smallskip

\item If $\varphi = \forall y \, \psi$, then
$\begin{aligned}[t]
\eta^{\vec{x}}_\varphi(R) \dfn \exists S\, &\forall \vec{x} ((\exists y R\vec{x}) \leftrightarrow (\exists y\, S\vec{x}{;}y))\land \eta_\psi^{\vec{x}{;}y}(S) \\
 \land\,&\forall \vec{x} (R\vec{x} \rightarrow \forall y\, S\vec{x}{;}y) \text{.}
\end{aligned}$
\end{itemize}

\begin{theorem}\label{thm:tl-to-so}
Let $\varphi \in \FOND$, let $\vv{x} \supseteq \Fr(\varphi)$ be a tuple of variables, and $T$ be a team in $\calA$ with domain $Y \supseteq \vv{x}$.
Then $(\calA,T) \vDash \varphi$ if and only if $\calA \vDash \eta^{\vec{x}}_\varphi(\vval{x}{T})$.
\end{theorem}

In order to prove the $\exists$ and $\forall$ cases, we require the next lemma.
\begin{lemma}[$\star$]\label{lem:supp-as-rest}
Let $T$ have domain $\vv{x}$ and $S$ have domain $\vv{x} \cup \{y\}$ (with possibly $y \in X$), and let $X' \dfn \vv{x} \setminus \{x\}$.
Then $\rest{T}{X'} = \rest{S}{X'}$ if and only if $\calA \vDash \pi(\vv{x}\langle T\rangle,\vv{x}{;}y\langle S\rangle)$, where $\pi(T,S) \dfn  \forall \vv{x}((\exists y T \vv{x}) \leftrightarrow (\exists y S \vv{x}{;}y))$.
\end{lemma}

\begin{proof}[Proof of Theorem]
Note that $(\calA,T) \vDash \varphi \Leftrightarrow (\calA,\rest{T}{\vv{x}}) \vDash \varphi$ by Proposition~\ref{prop:free-vars}, and $\vval{x}{T} = \vval{x}{\rest{T}{\vv{x}}}$.
For this reason, we can assume that $T$ has domain $\vv{x}$.
The proof is now by induction on $\varphi$.
\begin{itemize}
\item  If $\varphi$ is first-order, clearly $(\calA,T)\vDash \varphi$ iff $\calA \vDash \varphi(\vv{a})$ for all $\vv{a} \in \vval{x}{T}$ iff $\calA \vDash \eta^{\vv{x}}_\varphi(\vval{x}{T})$.

\item
If $\varphi = A_i(\vv{t})$ and $\delta_i \in \calD$ is a $k$-ary dependency, then $(\calA,T) \vDash A_i(\vv{t})$ iff $\calA \vDash \delta_i(\vval{t}{T})$.
We prove that this is again equivalent to $\calA \vDash \exists S \,\rho(R,S) \land \delta_i(S)$, where $R = \vval{x}{T}$ and $\rho(R,S) \dfn \forall \vv{z} \, (S\vv{z} \leftrightarrow (\exists \vv{x}\, R\vv{x} \land \vv{t} = \vv{z}\,))$.

It suffices to show that $\calA \vDash \rho(R,S)$ implies $S = \vval{t}{T}$.
Recall that $\vv{x} \cap \{z_1,\ldots,z_k\} = \emptyset$ and that the $z_i$ are pairwise distinct.
On that account, suppose $\calA \vDash \rho(R,S)$ and $\vv{a} = (a_1,\ldots,a_k)\in \calA^k$.
By definition of the formula, $\vv{a} \in S$ iff $\calA \vDash \exists \vv{x}\, R\vv{x} \land \vv{t} = \vv{a}$.
However, this is the case iff $\vv{t}\langle s \rangle = \vv{a}$ for some $s \in T$, \ie, $\vv{a} \in \vv{t}\langle T\rangle$.

\item
The cases $\varphi = \negg \psi$ and $\varphi = \psi \land \theta$ immediately follow by induction hypothesis.

\item
If $\varphi = \psi \lor \theta$, then by induction hypothesis, $(\calA,T) \vDash \varphi$ iff there are $S,U\subseteq T$ such that $T = S \cup U$ and $\calA \vDash \eta^{\vv{x}}_\psi(\vv{x}\langle S \rangle)\land \eta^{\vv{x}}_\theta(\vv{x}\langle U \rangle)$.
Let $R \dfn \vval{x}{T}$.

Then due to Proposition~\ref{prop:iso}, the above is equivalent to the existence of $P,Q \subseteq \calA^n$ such that $R = P \cup Q$ and $\calA \vDash \eta^{\vv{x}}_\psi(P)\land \eta^{\vv{x}}_\theta(Q)$, and consequently to $\calA \vDash \exists S \, \exists U\, \forall \vv{x}( R\vv{x} \leftrightarrow (S\vv{x} \lor U\vv{x})) \land \eta^{\vv{x}}_\psi(S) \land \eta^{\vv{x}}_\theta(U)$.

\item If $\varphi = \exists y\, \psi$, by Proposition~\ref{prop:def-of-supplement}, then $(\calA,T)\vDash \varphi$ iff $(\calA,S) \vDash \psi$ for some team $S$ with domain $\vv{x}\cup\{y\}$ such that $\rest{T}{X'}=\rest{S}{X'}$.
By Lemma~\ref{lem:supp-as-rest} and by induction hypothesis, this is the case iff $(\calA,\vval{x}{T}) \vDash \exists S\, \forall \vec{x}((\exists y R \vec{x}) \leftrightarrow (\exists y S \vv{x};y)) \land \eta^{\vec{x};y}_\psi(S)$.

\item The case $\varphi = \forall y\, \psi$ is proven analogously to $\exists$.
The additional clause $(R\vec{x} \rightarrow \forall y\, S\vec{x}{;}y)$ ensures that the supplementing function is constant and $f(s) = \dom\calA$. \qedhere
\end{itemize}
\end{proof}

\begin{definition}
Let $\calD = \{ \delta_1,\delta_2,\ldots\}$ be a set of dependencies.
$\calD$ is called \emph{p-uniform} if there is a polynomial time algorithm that for all $i$, when given $A_i\vv{t}$, computes $\delta_i$.
\end{definition}

\begin{corollary}\label{cor:fo-mc}
Let $\calD$ be a p-uniform set of dependencies.
Then $\MC(\FOND)$ is decidable in time $2^{n^{\bigO{1}}}$ and with $\size{\varphi}^\bigO{1}$ alternations, where $(\calA,T,\varphi)$ is the input.
\end{corollary}
\begin{proof}
We compute $\vv{x} \dfn \Fr(\varphi)$, $\eta^{\vv{x}}_\varphi$ and $\vval{x}{T}$ from $\varphi$ and $T$ in polynomial time.
Afterwards, we accept if and only if $\calA \vDash \eta^{\vv{x}}_\varphi(\vval{x}{T})$.

By Proposition~\ref{prop:so-mc}, the latter can be checked by an algorithm with $\size{\eta^{\vv{x}}_\varphi}$ alternations and time exponential in $(\calA,\vval{x}{T},\eta^{\vv{x}}_\varphi)$.
This corresponds to $\size{\varphi}^{\bigO{1}}$ alternations and runtime exponential in $(\calA,T,\varphi)$.
\end{proof}
 
\subsection*{Bounded quantifier rank or width}

The complexity of the model checking problem of team logic significantly drops if either the number of variables or the quantifier rank is bounded by an arbitrary constant.
To prove this, we require a special fragment of $\SO$ that corresponds to such fragments.
We call this fragment \emph{sparse second-order logic}, which uses \emph{sparse quantifiers} $\exists^p$ and $\forall^p$:
\begin{align*}
\!\!\!\!\!\!\!\!\!(\calA,\calJ) \vDash \exists^p P\, \psi &\Leftrightarrow \exists R\subseteq \calA^{\arity{P}} \text{ such that }\size{R} \leq p(\size{\calA}) \text{ and }(\calA,\calJ^P_R) \vDash \psi\text{,}\\
\!\!\!\!\!\!\!\!\!(\calA,\calJ) \vDash \forall^p P\, \psi &\Leftrightarrow \forall R \subseteq \calA^{\arity{P}} \text{ such that }\size{R} \leq p(\size{\calA}) \text{ it holds }(\calA,\calJ^P_R) \vDash \psi\text{,}
\end{align*}
where $p\colon \N \to \N$ and $\size{\calA} \dfn \size{\!\dom\calA} + \sum_{X \in \tau}\size{X^\calA}$.
In other words, all quantified relations are bounded in their size relative to the underlying $\tau$-structure.
For obvious reason, there are no sparse function quantifiers.

The logic $\SO[p]$ is now defined as $\SO$, but with only $\exists^p$ and $\forall^p$ as permitted second-order quantifiers.

Consider the case where $p$ is bounded by a polynomial.
The interpretation of each quantified relation then contains at most $\size{\calA}^{\bigO{1}}$ tuples.
Consequently, on $\SO[p]$-formulas, the recursive model checking algorithm from Proposition~\ref{prop:so-mc} then runs in alternating polynomial time instead of exponential time.

\begin{corollary}\label{cor:bounded-so-pspace}
If $p$ is bounded by a polynomial, then
$\MC(\SO[p])$ is decidable in polynomial time with $\size{\alpha}$ alternations, where $(\calA,\calJ,\alpha)$ is the input.
\end{corollary}

Remarkably, the translation from team logic with bounded width or quantifier rank to $\SO$ takes place in this fragment.
For $p \colon \N \to \N$, define the $\SO[p]$-formula $\zeta^{\vv{x},p}_\varphi$ similarly to $\eta^{\vv{x}}_\varphi$, but with all second-order quantifiers replaced by $\exists^{p}$.

Intuitively, every quantified relation in $\eta^{\vv{x}}_\varphi$ is either dividing or supplementing an existing team resp.\ relation.
For this reason, they grow at most by a factor of $\size{\!\dom\calA}$ for every quantifier.

\begin{theorem}\label{thm:sparse-translation}
Let $\varphi \in \FOND$, let $\vv{x} \supseteq \Fr(\varphi)$ be a tuple of variables, and $T$ be a team in $\calA$ with domain $Y \supseteq \vv{x}$.
If $p(n) \geq \size{T}\cdot n^{\qr(\varphi)}$ or $p(n) \geq n^{\wi(\varphi)}$, then $(\calA,T)\vDash \varphi$ if and only if $\calA \vDash \zeta^{\vv{x},p}_\varphi(\vval{x}{T})$.
\end{theorem}

We distinguish two cases.

\begin{proof}[Proof for $p(n) \geq \size{T}\cdot n^{\qr(\varphi)}$]

Assume $\calA$, $T$ as above, let $m \dfn \qr(\varphi)$ and $p(n) \geq n^m$.

The idea of the proof is to show that $\eta^{\vv{x}}_\varphi$ and $\zeta^{\vv{x},p}_\varphi$ agree on $(\calA,\calJ)$ for all "sufficiently sparse" $\calJ$ (cf.\ Theorem~\ref{thm:tl-to-so}).

Formally, let $\ell \leq m$ and let $(\calA,\calJ)$ be a second-order interpretation such that $\size{\calJ(R)}\leq \size{T}\cdot \size{\calA}^{\ell}$ for all relations $R \in \dom \calJ$.
Then we prove for all $\varphi \in \FOND$ with $\qr(\varphi) \leq m - \ell$ and tuples $\vv{x} \supseteq \Fr(\varphi)$ that it holds $(\calA,\calJ) \vDash \eta^{\vv{x}}_\varphi \Leftrightarrow (\calA,\calJ) \vDash \zeta^{\vv{x},p}_\varphi$ .

For $\ell = 0$, this yields the theorem, since $\size{\vval{x}{T}} \leq \size{T} \cdot \size{\calA}^0$.

The proof is by induction on $\varphi$.
We distinguish the following cases.
\begin{itemize}
\item If $\varphi \in \FO$, then $\eta^{\vv{x}}_\varphi = \zeta^{\vv{x},p}_\varphi$.
\item If $\varphi = \negg \psi$ or $\varphi = \psi \land \theta$, then the inductive step is clear.
\item If $\varphi = A_i\vv{t}$ for a $k$-ary dependency $\delta_i \in \calD$, then
$\zeta^{\vv{x},p}_\varphi(R) = \exists^p S \,\rho(R,S)$ and $\eta^{\vv{x}}_\varphi(R) = \exists S\, \rho(R,S)$, where $\rho(R,S) = \forall \vv{z} \, (S\vv{z} \leftrightarrow (\exists \vv{x}\, R\vv{x} \land \vv{t} = \vv{z}\,)) \land \delta_i(S)$.
We show that $\calA \vDash \eta^{\vv{x}}_\varphi(R)$ implies $\calA \vDash \zeta^{\vv{x},p}(R)$.
The other direction is trivial.

On that account, suppose $\calA \vDash \rho(R,S)$ for some $S \subseteq \calA^k$.
We construct an injection $f \colon S  \to R$, which implies $\size{S} \leq \size{R}$.
Consequently, $\calA \vDash \exists^p S\,\rho(R,S)$, as by assumption, $\size{R} \leq \size{T}\cdot \size{\calA}^\ell \leq \size{T}\cdot \size{\calA}^m \leq p(\size{\calA})$.

For every $\vv{a} \in S$, define $f(\vv{a})$ as some $\vv{b} \in R$ such that $\vv{t}\langle\{\vv{x}\mapsto \vv{b}\}\rangle = \vv{a}$.
By $\rho(R,S)$, such $\vv{b}$ must exist.
Clearly, $f$ is injective.
\item If $\varphi = \psi \lor \theta$, then $\zeta^{\vv{x},p}_\varphi(R) = \exists^p S\, \exists^p U \, \rho$ and $\eta^{\vv{x}}_\varphi(R) = \exists S \,\exists U \, \rho'$, where
\begin{align*}
\rho(R,S,U) &= \forall \vec{x} (R\vv{x} \leftrightarrow (S\vv{x} \lor U\vv{x}))\land \zeta^{\vv{x},p}_\psi(S) \land \zeta^{\vv{x},p}_\theta(U)\text{,}\\
\rho'(R,S,U) &= \forall \vec{x} (R\vv{x} \leftrightarrow (S\vv{x} \lor U\vv{x}))\land \eta^{\vv{x}}_\psi(S) \land \eta^{\vv{x}}_\theta(U)\text{.}
\end{align*}
Suppose $\size{R} \leq \size{T}\cdot\size{\calA}^\ell$ and $\qr(\varphi) \leq m - \ell$.
Clearly $\qr(\psi) = \qr(\theta) = \qr(\varphi)$.

Let $\calA \vDash \eta^{\vv{x}}_\varphi(R)$, \ie, $\calA \vDash \rho'(R,S,U)$ for some $S,U \subseteq \calA^{\size{\vv{x}}}$.
But $\rho'$ enforces that $\size{S},\size{U} \leq \size{R}$.
Since $\size{R} \leq \size{T}\cdot\size{\calA}^\ell$ by assumption, we can apply the induction hypothesis to $\eta^{\vv{x}}_\psi(S)$ and $\eta^{\vv{x}}_\theta(U)$ and obtain $\calA \vDash \rho(R,S,U)$.
Since in particular $\size{S},\size{U} \leq p(\size{\calA})$, we conclude $\calA \vDash \exists^p S\, \exists^p U\,\rho$.
The other direction is similar.
\smallskip

\item If $\varphi = \exists y\, \psi$, then $\zeta^{\vv{x},p}_\varphi(R) = \exists^p S\,\rho(R,S)$ and $\eta^{\vv{x}}_\varphi(R) = \exists S \,\rho'(R,S)$, where
\begin{align*}
\rho(R,S) &= \forall \vec{x} ((\exists y R\vec{x}) \leftrightarrow (\exists y\, S\vec{x}{;}y)) \land \zeta_\psi^{\vec{x}{;}y,p}(S)\text{,}\\
\rho'(R,S) &= \forall \vec{x} ((\exists y R\vec{x}) \leftrightarrow (\exists y\, S\vec{x}{;}y)) \land \eta_\psi^{\vec{x}{;}y}(S)\text{.}
\end{align*}
Suppose $\size{R} \leq \size{T}\cdot\size{\calA}^\ell$ and $\qr(\varphi) \leq m - \ell$.
We show that $\calA \vDash \eta^{\vv{x}}_\varphi(R)$ implies $\calA \vDash \zeta^{\vv{x},p}_{\varphi}(R)$.
The other direction is again similar.

Assuming $\calA \vDash \eta^{\vv{x}}_\varphi(R)$, there exists $S \subseteq \calA^{\size{\vv{x};y}}$ such that $\calA \vDash \rho'(R,S)$.
As a first step, we erase unnecessary elements from $S$.
Note that $S$ occurs in $\rho'$ only in atomic formulas $S\vv{x}{;}y$, \ie, with a fixed argument tuple $\vv{x}{;}y$.
Let $(v_1,\ldots,v_r) \dfn \vv{x}{;}y$.
If now $v_i = v_j$ for some $1 \leq i < j \leq r$, then every tuple $(a_1,\ldots,a_r)$ with $a_i \neq a_j$ can be safely deleted from $S$.
Formally, if $S^* \dfn \vv{x}{;}y\langle V \rangle \cap S$, where $V = \{ s\colon \vv{x}\cup\{y\} \to \calA \}$, then $\calA \vDash \rho'(R,S)$ if and only if $\calA \vDash \rho'(R,S^*)$, which can be showed by straightforward induction.

\smallskip

Note that $\qr(\psi) = \qr(\varphi) - 1 \leq m - (\ell + 1)$.
Consequently, to apply the induction hypothesis, we prove
$\size{S^*} \leq \size{R} \cdot \size{\calA}\leq \size{T}\cdot \size{\calA}^{\ell+1}$ by presenting some injective $f \colon S^* \to R \times \calA$.

If $y \notin \vv{x}$, let $f$ be the identity, as $\rho'$ ensures that $(\vv{a},b)\in S^*$ implies $\vv{a}\in R$.
However, if $y \in \vv{x}$, then we define $f(\vv{a})$ as follows.
By construction, $\vv{a} \in S^*$ equals $\vv{x}\langle s \rangle$ for some $s \colon \vv{x}\to\calA$.
Again by $\rho'$, there is $\hat{s}\colon \vv{x}\to \calA$ such that $\vv{x}\langle \hat{s} \rangle \in R$ and $s = \hat{s}^y_{s(y)}$.
Let now $f(\vv{a}) \dfn (\vval{x}{\hat{s}}, s(y))$.
Then $f$ is injective.

Hence, by induction hypothesis, we can replace $\eta^{\vv{x}}_\varphi$ by $\zeta^{\vv{x},p}_\varphi$ and obtain $\calA \vDash \rho(R,S^*)$.
Since $\size{S^*} \leq \size{\calA}^{\ell + 1} \leq p(\size{\calA})$, we obtain $\calA \vDash \exists^p S\, \rho(R,S)$.
\item The case $\varphi = \forall y\, \psi$ is proven similarly to $\varphi = \exists y\, \psi$.\qedhere
\end{itemize}
\end{proof}

\begin{proof}[Proof for $p(n) \geq n^{\wi(\varphi)}$]
We can apply the same argument as in the $\exists$-case of the previous proof.
Suppose $S$ is a second-order variable.
Then $S$ appears in $\eta^{\vv{x}}_\varphi$ only in atomic formulas of the form $S\vv{t}$ for a fixed $\vv{t}$.
Accordingly, it suffices to let $\exists S$ range over subsets of $\vv{t}\langle V \rangle$, where $\vv{y} \dfn \Var(\vv{t})$ and $V \dfn \{ s \colon \vv{y} \to \calA \}$.

(We consider terms $\vv{t}$ instead of only variables to account for the translations of dependencies, where $S$ can have terms as arguments.)

Since $\vv{y}$ contains at most $\wi(\varphi)$ distinct variables, $\size{V} \leq \size{\calA}^{\wi(\varphi)} \leq p(\size{\calA})$.
Consequently, every second-order quantifier $\exists S$ can be replaced by $\exists^p S$, which implies $\calA \vDash \eta^{\vv{x}}_\varphi(\vval{x}{T}) \Leftrightarrow \calA \vDash \zeta^{\vv{x},p}_\varphi(\vval{x}{T})$.
\end{proof}

\begin{corollary}\label{cor:mcfo2-upper1}
Let $\calD$ be a p-uniform set of dependencies and $m < \omega$.
$\MC(\FONKD{m}{\omega})$ and $\MC(\FONKD{\omega}{m})$ are then decidable in polynomial time with $\size{\varphi}^\bigO{1}$ alternations, where $(\calA,T,\varphi)$ is the input.
\end{corollary}
\begin{proof}
Analogously to Corollary~\ref{cor:fo-mc}, we reduce to $\MC(\SO[p])$ in polynomial time, where $p(n) \dfn n^{m+1}$.
Accordingly, assume an input $(\calA,T,\varphi)$ and that either $\wi(\varphi) \leq m$ or $\qr(\varphi) \leq m$.
Then the input is mapped to $(\calA,\vval{x}{T},\zeta^{p,\vv{x}}_\varphi)$, where $\vv{x} = \Fr(\varphi)$.

Now, if $\wi(\varphi) \leq m$, then $(\calA,T) \vDash \varphi$ if and only if $\calA \vDash \zeta^{p,\vv{x}}_\varphi(\vval{x}{T})$ by Theorem~\ref{thm:sparse-translation}.

Otherwise, $\qr(\varphi) \leq m$.
In that case, \wloss $\size{T} \leq \size{\calA}$ (if necessary, pad $\calA$ with a dummy predicate).
Then $\size{T}\cdot \size{\calA}^m \leq \size{\calA}^{m+1} = p(\size{\calA})$.
As a consequence, we can again apply Theorem~\ref{thm:sparse-translation} and conclude that $(\calA,T) \vDash \varphi$ if and only if $\calA \vDash \zeta^{p,\vv{x}}_\varphi(\vval{x}{T})$.
\end{proof}
 
\section{From FO$^2$($\negg$) to FO$^2$: Upper bounds for satisfiability}

\label{sec:upper-sat}

In this section, we prove that the satisfiability problem of $\FONN$ is decidable, and in fact, complete for $\TOWERPOLY$.

Our approach is to establish a finite model property for $\FONN$.
However, instead of attacking $\FONN$ directly, we reduce the logic to $\FO^2$, which has the exponential model property \cite{fo2}.
As a first step, we expand $\FONN$-formulas into disjunctive normal form with respect to $\land$ and $\negg$, using the laws depicted in the next lemma.

We use the abbreviations $\varphi \ovee \psi \dfn \negg(\negg \varphi \land \negg \psi)$ (Boolean disjunction) and $\E\beta \dfn \negg \neg \beta$ ("at least one assignment in the team satisfies $\beta$").

\begin{lemma}\label{lem:axioms}
The following laws hold for $\FON$:
\begin{alignat}{2}
&\alpha \land \bigwedge_{i=1}^n \E\beta_i &&\equiv \bigvee_{i=1}^n(\alpha \land \E\beta_i) \label{eq:flat-dist1} \\
&\bigvee_{i=1}^n(\alpha_i \land \E \beta_i) &&\equiv \Big(\bigvee_{i=1}^n \alpha_i\Big) \land \bigwedge_{i=1}^n \E(\alpha_i\land \beta_i) \label{eq:flat-dist2}\\
&(\vartheta_1 \ovee \vartheta_2) \lor \vartheta_3 &&\equiv (\vartheta_1 \lor \vartheta_3) \ovee (\vartheta_2 \lor \vartheta_3) \label{eq:ovee-lor-dist1}\\
&\vartheta_1 \lor (\vartheta_2 \ovee \vartheta_3) &&\equiv (\vartheta_1 \lor \vartheta_2) \ovee (\vartheta_1 \lor \vartheta_3) \label{eq:ovee-lor-dist2}\\
&\exists x\, (\vartheta_1 \ovee \vartheta_2) &&\equiv (\exists x\,\vartheta_1) \ovee (\exists x\, \vartheta_2) \label{eq:ovee-ex-dist}\\
&\exists x\, (\vartheta_1 \vee \vartheta_2) &&\equiv (\exists x\,\vartheta_1) \vee (\exists x\, \vartheta_2) \label{eq:lor-ex-dist}\\
&\exists x\, (\alpha \land \E \beta) &&\equiv (\exists x\, \alpha) \land \E\,\exists x\,(\alpha \land \beta) \label{eq:land-ex-dist}\\
&\forall x\, (\vartheta_1 \land \vartheta_2) &&\equiv (\forall x \, \vartheta_1) \land (\forall x \, \vartheta_2) \label{eq:land-all-dist}\\
&\forall x \, \negg \vartheta &&\equiv \negg \forall x \, \vartheta \label{eq:neg-all-dist}
\end{alignat}
\end{lemma}
\begin{proof}
For (\ref{eq:flat-dist1}), (\ref{eq:flat-dist2}) and (\ref{eq:lor-ex-dist}), see \cite[Lemma 4.13, 4.14 and D.1]{axiom-fo}, respectively.
For (\ref{eq:ovee-lor-dist1})--(\ref{eq:ovee-ex-dist}), see \cite[Proposition 5]{abramsky_strongly_2016}.
For (\ref{eq:land-all-dist})--(\ref{eq:neg-all-dist}), see \cite[Chapter 8]{vaananen_dependence_2007}.

For (\ref{eq:land-ex-dist}), the direction "$\vDash$" is clear, as $\alpha\land \E\beta$ implies both $\alpha$ and $\E(\alpha \land \beta)$.
For the converse direction, suppose $(\calA,T) \vDash \exists x \alpha$ and $(\calA,\hat{s}) \vDash \exists x (\alpha \land \beta)$ for some $\hat{s} \in T$.
Then there are $f \colon T \to \pow{\calA}\setminus \{\emptyset\}$ and $b \in \calA$ such that $(\calA,T^x_f) \vDash \alpha$ and $(\calA,\hat{s}^x_b) \vDash \alpha \land \beta$.
Define $g(\hat{s}) = f(\hat{s}) \cup \{b\}$ and $g(s) = f(s)$ for $s \in T \setminus \{\hat{s}\}$.
Then $T^x_g = T^x_f \cup \{ s^x_b \}$.
Consequently, $(\calA,T^x_g) \vDash \alpha \land \E\beta$, hence $(\calA,T) \vDash \exists x\, (\alpha \land \E \beta)$.
\end{proof}

\begin{lemma}\label{lem:expansion}
Every $\tau\text{-}\FONK{n}{k}$-formula $\varphi$ is equivalent to a formula of the form
\[
\psi \dfn \bigovee_{i = 1}^n  \left( \alpha_i \land \bigwedge_{j = 1}^{m_i} \E \beta_{i,j}\right)
\]
such that $\{\alpha_1,\ldots,\alpha_n,\beta_{1,1},\ldots,\beta_{n,m_n}\} \subseteq \tau\text{-}\FO^n_k$ and $\size{\psi} \leq \exp_{\bigO{\size{\varphi}}}(1)$.
\end{lemma}
In the following, \emph{disjunctive normal form} (DNF) refers to formulas in the above form.
\begin{proof}
Proof by induction on $\varphi$.
\begin{itemize}
\item If $\varphi$ is a Boolean combination of $\FO^n_k$-formulas (over $\negg$ and $\land$), then we obtain a DNF of size $\leq \size{\varphi}\cdot 2^{\size{\varphi}}$ similarly as for ordinary propositional logic.
\item If $\varphi = \vartheta_1 \lor \vartheta_2$ for $\vartheta_1$, $\vartheta_2$ in DNF, then
\begin{align*}
\varphi = \;&\bigovee_{i = 1}^n
\left( \alpha_i \land \bigwedge_{j =1}^{m_i} \E \beta_{i,j}\right) \;
\lor
\;\bigovee_{i = 1}^{k} \left( \gamma_i \land \bigwedge_{j = 1}^{\ell_i} \E \delta_{i,j} \right) \\
\equiv \;&\bigovee_{i = 1}^n
\bigvee_{j=1}^{m_i} (\alpha_i \land \E \beta_{i,j}) \;
\lor
\;\bigovee_{i = 1}^{k}  \bigvee_{j=1}^{\ell_i}( \gamma_i \land \E \delta_{i,j} )\tag{Lemma~\ref{lem:axioms}, (\ref{eq:flat-dist1})}\\
\equiv \;& \bigovee_{\substack{1 \leq i_1 \leq n\\1\leq i_2 \leq k}} \;\; \bigvee_{j=1}^{m_{i_1}} (\alpha_{i_1} \land  \E \beta_{i_1,j}) \lor \bigvee_{j=1}^{\ell_{i_2}} (\gamma_{i_2} \land \E \delta_{i_2,j} ) \tag{Lemma~\ref{lem:axioms}, (\ref{eq:ovee-lor-dist1}) and  (\ref{eq:ovee-lor-dist2})}\\
\equiv \;& \bigovee_{i = 1}^{n\cdot k}  \; \bigvee_{j=1}^{o_i} (\mu_{i,j} \land  \E \nu_{i,j}) \tag{where $\mu_{i,j},\nu_{i,j} \in \FO^n_k$}\\
\equiv \;& \bigovee_{i = 1}^{n\cdot k} \Big(\bigvee_{j=1}^{o_i} \mu_{i,j}\Big) \land \bigwedge_{j = 1}^{o_i} \E(\mu_{i,j} \land \nu_{i,j})\text{,}\tag{Lemma~\ref{lem:axioms}, (\ref{eq:flat-dist2})}
\end{align*}
where the final DNF has size polynomial in $\size{\vartheta_1} + \size{\vartheta_2} \leq \size{\varphi}$.
\item If $\varphi = \exists x \, \vartheta$ for $\vartheta$ in DNF, then
\begin{align*}
\varphi \equiv\;& \exists x \,\bigovee_{i = 1}^n \bigvee_{j=1}^{m_i} (\alpha_i \land  \E \beta_{i,j}) \tag{Lemma~\ref{lem:axioms}, (\ref{eq:flat-dist1})}\\
\equiv\; & \bigovee_{i = 1}^n \bigvee_{j=1}^{m_i} \exists x \,(\alpha_i \land  \E \beta_{i,j}) \tag{Lemma~\ref{lem:axioms}, (\ref{eq:ovee-ex-dist}) and (\ref{eq:lor-ex-dist})}\\
\equiv\; & \bigovee_{i = 1}^n \bigvee_{j=1}^{m_i} \Big( (\exists x \,\alpha_i) \land  \E\exists x \,(\alpha_i \land \beta_{i,j}) \Big) \tag{Lemma~\ref{lem:axioms}, (\ref{eq:land-ex-dist})}\\
\equiv \; & \bigovee_{i = 1}^n \Big(\bigvee_{j=1}^{m_i} \exists x\, \alpha_i \Big) \land \bigwedge_{j=1}^{m_i} \E\exists x\, (\alpha_i\land \beta_{i,j}))\text{,}\tag{Lemma~\ref{lem:axioms}, (\ref{eq:flat-dist2})}
\end{align*}
which is again a DNF of polynomial size.
\item Finally, the $\forall$ case is due to (\ref{eq:land-all-dist}) and (\ref{eq:neg-all-dist}) of Lemma~\ref{lem:axioms}.\qedhere
\end{itemize}
\end{proof}

\begin{theorem}\label{thm:fo2-modelsize}
If $\tau$ is a relational vocabulary, then every satisfiable $\varphi \in \tau\text{-}\FONN$ has a model of size $\exp_{\bigO{\size{\varphi}}}(1)$.
\end{theorem}
\begin{proof}
Let $\varphi \in \tau\text{-}\FONN$ be satisfiable.
By Lemma~\ref{lem:expansion}, $\varphi$ is equivalent to a DNF of size $\exp_{\bigO{\size{\varphi}}}(1)$, which then has at least one satisfiable disjunct
\[
\psi = \alpha \land \bigwedge_{i=1}^m \E \beta_i\text{,}
\]
where $\{\alpha,\beta_1,\ldots,\beta_m\} \subseteq \tau\text{-}\FO^2$ and \wloss $\Var(\psi) \subseteq \{x,y\}$.
Let $(\calA, T)$ be a model of $\psi$.
We show that $\psi$, and hence $\varphi$, has a model $(\calB,T')$ of size $2^{\bigO{\size{\psi}^2}}$.

Observe that $\calA$ satisfies the classical $\FO^2$-sentence
\[
\gamma \dfn \bigwedge_{i=1}^m \exists x\,\exists y\, \alpha \land \beta_i\text{,}
\]
since for every $i$, $(\calA,T)\vDash \alpha \land\E \beta_i$, which implies that there is $s \in T$ such that $(\calA,s) \vDash \alpha \land \beta_i$, \ie, $\calA \vDash \exists x \exists y \alpha \land \beta_i$.

By \citet{fo2}, if $\delta$ is a satisfiable $\FO^2$-formula over a relational vocabulary, then $\delta$ has a model of size $2^{\bigO{\size{\delta}}}$ (in Tarski semantics).
Let $\calB$ the corresponding model of $\gamma$ of size $2^{\bigO{\size{\gamma}}} \subseteq 2^{\bigO{\size{\psi}^2}}$.
Then for every $i$ there exists $\hat{s}_i \colon \{x,y\} \to \calB$ such that $(\calB,\hat{s}_i) \vDash \alpha \land \beta_i$.
Consequently, $(\calB,\{\hat{s}_1,\ldots,\hat{s}_m\}) \vDash \psi$.
\end{proof}

\begin{corollary}\label{cor:fo2-upperbound}
If $\tau$ is a relational vocabulary, then $\SAT(\tau\text{-}\FONN)$ and $\VAL(\tau\text{-}\FONN)$ are in $\TOWERPOLY$.
\end{corollary}
\begin{proof}
By Corollary~\ref{cor:mcfo2-upper1}, model checking for $\FONN$ is possible in alternating polynomial time, and hence in deterministic exponential time.
In order to obtain a $\TOWERPOLY$-algorithm for $\SAT(\tau\text{-}\FONN)$, given a formula $\varphi \in \tau\text{-}\FONN$, we can now iterate over all interpretations $(\calA,T)$ of size $\exp_{\bigO{\size{\varphi}}}(1)$ and perform model checking on $(\calA,T,\varphi)$.
\end{proof}
 \section{From MTL to FO$^2$($\negg$): Lower bounds}

\label{sec:lower}

For the lower bounds, we reduce from the model checking and the satisfiability problem of $\MTL$ and its fragments.
As a part of the reduction, in this section we present a semantics-preserving translation from $\MTL$ to $\FONN$.

The \emph{standard translation} embeds modal logic $\ML$ into $\FO^2$ with the relational vocabulary $\tau_M = (R,P_1,P_2,\ldots)$, where $\arity{R} = 2$ and $\arity{P_i} = 1$.
The standard translation of an $\ML$-formula is denoted by $\stx(\varphi)$ resp.\ $\sty(\varphi)$ and defined recursively:
\begin{align*}
\stx(p_i)&\dfn P_ix \;\text{ for }p_i \in \PS & 
\stx(\neg \alpha) &\dfn \neg \stx(\alpha)\\
\stx(\Box \alpha) & \dfn \forall y \, Rxy \imp \sty(\alpha) &
\stx(\alpha \land \beta)& \dfn \stx(\alpha) \land \stx(\beta)\\
\stx(\Diamond \alpha) & \dfn \exists y \, Rxy \land \sty(\alpha) & 
\stx(\alpha \lor \beta)& \dfn \stx(\alpha) \lor \stx(\beta)\text{,}
\end{align*}
with $\sty(\alpha)$ defined symmetrically via $\stx(\alpha)$.
The corresponding \emph{first-order interpretation} of a Kripke structure $\calK = (W,\calR,V)$ is the $\tau_M$-structure $\calA(\calK)$ such that $\dom \calA(\calK) = W$, $R^{\calA(\calK)} = \calR$ and $P_i^{\calA(\calK)} = V(p_i)$.
For a world $w$, define $w^x \colon \{ x \} \to W$ by $w^x(x) = w$.

The next theorem is standard (see, \eg, \citet{blackburn_modal_2001}):

\begin{theorem}\label{thm:st-trans-ml}
Let $(\calK,w)$ be a pointed Kripke structure and $\alpha \in \ML$.
Then
$(\calK,w) \vDash \alpha$ if and only if $(\calA(\calK),w^x) \vDash \stx(\alpha)$.
\end{theorem}

In this section, we lift this result to team semantics and translate $\MTL$ to $\FONN$.
On the model side, the first-order interpretation of a team $T$ in a Kripke structure becomes $T^x \dfn \Set{ w^x | w \in T}$.

The standard translation for $\MTL$ now extends that of $\ML$ by the final $\negg$-case and also implements the $\Box$-case with  $\hook$ instead of $\imp$:
\begin{align*}
\stx(p_i)&\dfn P_ix \;\text{ for }p_i \in \PS & 
\stx(\neg \varphi) &\dfn \neg \stx(\varphi)\\
\stx(\Box \varphi) & \dfn \forall y \, Rxy \hook \sty(\varphi) &
\stx(\varphi \land \psi)& \dfn \stx(\varphi) \land \stx(\psi)\\
\stx(\Diamond \varphi) & \dfn \exists y \, Rxy \land \sty(\varphi) & 
\stx(\varphi \lor \psi)& \dfn \stx(\varphi) \lor \stx(\psi)\\
\stx(\negg \varphi) & \dfn \negg \stx(\varphi)\text{,} & &
\end{align*}
with $\sty(\varphi)$ again defined symmetrically.
It is easy to see that the "classical" translation of $\Box\varphi$ to $\forall y \, Rxy \to \sty(\varphi) = \forall y \, (\neg Rxy \lor \sty(\varphi))$ is unsound under team semantics, and using $\forall y\, Rxy \hook \sty(\varphi) = \forall y\, (\neg Rxy \lor Rxy \land \sty(\varphi))$ is indeed required.

\begin{theorem}\label{thm:translation}
For all Kripke structures with team $(\calK,T)$, all $\varphi \in \MTL$ and any first-order variable $x$ it holds $(\calK,T) \vDash \varphi$ if and only if $(\calA(\calK),T^x)\vDash \stx(\varphi)$.
\end{theorem}
\begin{proof}
Proof by induction on $\varphi$.
Here, we omit $\calK$ and $\calA(\calK)$ and write, \eg, $T \vDash \varphi$.

\begin{itemize}
\item $\varphi \in \ML$: We have $(\calK,T)\vDash \varphi$ iff $\forall w\in T\colon (\calK,w)\vDash \varphi$ by definition of the semantics of $\MTL$, which by Theorem~\ref{thm:st-trans-ml} is equivalent to $\forall w^x \in T^x\colon (\calA(\calK),w^x)\vDash \stx(\varphi)$.
However, as $\stx(\varphi) \in \FO$, the latter is equivalent to $(\calA(\calK),T^x)\vDash \stx(\varphi)$ by the semantics of $\FON$.

\item $\varphi = \psi \land \vartheta$ and $\varphi = \negg \psi$ are clear.
\item $\varphi = \psi \lor \theta$:
Suppose $T \vDash  \psi \lor \theta$.
Then $T = S \cup U$ such that $S \vDash \psi$ and $U \vDash \theta$.
By induction hypothesis, $S^x \vDash \stx(\psi)$ and $U^x \vDash \stx(\theta)$.
As $S \cup U = T$, clearly $S^x \cup U^x = T^x$.
As a consequence, $T^x \vDash \stx(\psi) \lor \stx(\theta) = \stx(\psi \lor \theta)$.

\smallskip

For the other direction, suppose $T^x \vDash \stx(\psi\lor\theta) = \stx(\psi) \lor \stx(\theta)$ by the means of some subteams $S' \cup U' = T^x$ such that $S' \vDash \stx(\psi)$ and $U' \vDash \stx(\theta)$.
Then $S' = S^x$ and $U' = U^x$ for some suitably chosen $S,U \subseteq T$.
By induction hypothesis, $S\vDash \psi$ and $U\vDash \theta$.
In order to prove $T\vDash \psi\lor\theta$, it remains to show $T \subseteq S \cup U$.
For this purpose, let $w \in T$.
As then $w^x \in T^x$, as least one of $w^x \in S'$ or $w^x \in U'$ holds.
But then $w \in S$ or $w \in U$.
\item $\varphi = \Box\psi$:
We define subteams $S$ and $U$ of the duplicating team $(T^x)^y_W$ as follows: $S$ contains all "outgoing edges": $S \dfn \{ s \in (T^x)^y_W \mid s(y) \in Rs(x) \}$.
On the other hand, $U$ contains all "non-edges":
$U \dfn \{ s \in (T^x)^y_W \mid s(y) \notin Rs(x) \}$.
Then clearly $(T^x)^y_W = S \cup U$, $S \vDash Rxy$ and $U \vDash \neg Rxy$.
Moreover, the above division of $(T^x)^y_W$ into $S$ and $U$ is the \emph{only} possible splitting of $(T^x)^y_W$ such that $S \vDash Rxy$ and $U \vDash \neg Rxy$.

By induction hypothesis, clearly $T\vDash \Box\psi \Leftrightarrow (RT)^y \vDash \sty(\psi)$.
Moreover, by the above argument, $T^x \vDash \stx(\Box \psi) \Leftrightarrow S \vDash \sty(\psi)$.
Consequently, it suffices to show that $(RT)^y$ and $S$ agree on $\sty(\psi)$.
This follows from Proposition~\ref{prop:free-vars}, since
\begin{align*}
(RT)^y &= \{ s\colon \{y\} \to W \mid \exists w \in T : s(y) \in Rw  \}\\
&= \{ \rest{s}{y} \mid s \in (T^x)^y_W \text{ and }s(y) \in Rs(x) \} = \rest{S}{y}\text{.}
\end{align*}
\item $\varphi = \Diamond\psi$:
Suppose $T \vDash \Diamond \psi$, \ie, $S \vDash \psi$ for some successor team $S$ of $T$.
By induction hypothesis, $S^y \vDash \sty(\psi)$.
In order to prove $T^x \vDash \exists y \, Rxy \land \sty(\psi)$, we define a supplementing function $f : T^x \to \pow{W}\setminus \{\emptyset\}$ such that $(T^x)^y_f \vDash Rxy \land \sty(\psi)$.

Let $f(w^x) \dfn Rw \cap S$.
Then $f(w^x)$ is non-empty for each $w$, as $S$ is a successor team.
Moreover, $(T^x)^y_f \vDash Rxy$.
It remains to show that $(T^x)^y_f\vDash \sty(\psi)$ follows from $S^y \vDash \sty(\psi)$.
Here, we combine Proposition~\ref{prop:rel} and \ref{prop:free-vars}, since
\[
y\langle S^y \rangle = S = \bigcup_{w \in T} Rw \cap S = \bigcup_{w^x \in T^x} f(w^x) = \{ s(y) \mid s \in (T^x)^y_f\} = y\langle (T^x)^y_f\rangle\text{.}
\]

For the other direction, suppose $T^x \vDash \exists y\, Rxy \land \sty(\psi)$ by the means of a supplementing function $f : T^x \to \pow{W}\setminus \{\emptyset\}$ such that $(T^x)^y_f \vDash Rxy \land \sty(\psi)$.

We define $S \dfn \bigcup_{w \in T} f(w^x)$, and first prove that it is a successor team of $T$, \ie, that every $v\in S$ has a predecessor in $T$ and that every $w \in T$ has a successor in $S$.

Let $v \in S$.
Then there exists $w \in T$ such that $v \in f(w^x)$.
As a consequence, the assignment $s$ given by $s(x) = w$ and $s(y) = v$ is in $(T^x)^y_f$, and hence satisfies $Rxy$.
In other words, $v$ has a predecessor in $T$.
Conversely, if $w \in T$, then $f(w^x)$ is non-empty, \ie, contains an element $v$.
As before, $v$ is a successor of $w$.
Since $v \in f(w^x)$, $v \in S$, so $w$ has a successor in $S$.

By a similar argument as above, $y\langle (T^x)^y_f\rangle = S = y\langle S^y\rangle$, hence $S^y \vDash \sty(\psi)$, and consequently $S \vDash \psi$ by induction hypothesis.
\qedhere
\end{itemize}
\end{proof}
 
Next, we prove that the standard translation carries several complexity lower bounds into the first-order setting.

\begin{lemma}\label{lem:fo-mc-hardness-nk}
$\MC(\tau\text{-}\FONK{1}{0})$ is $\PSPACE$-hard if $\tau$ contains infinitely many predicates.
\end{lemma}
\begin{proof}
We reduce from $\MC(\PTL)$ (see Theorem~\ref{thm:mc-ptl}) via $(\calK,T,\varphi) \mapsto (\calA(\calK),T^x,\stx(\varphi))$.
(\Wloss $\tau$ contains unary predicates $P_0, P_1, \ldots$; otherwise they are easily simulated by predicates of higher arity).

It is easy to see that $\stx(\varphi)$ is quantifier-free and contains only the variable $x$.
Moreover, by Theorem~\ref{thm:translation}, $(\calK,T) \vDash \varphi\Leftrightarrow (\calA(\calK),T^x) \vDash \stx(\varphi)$.
\end{proof}

\begin{lemma}\label{lem:fo-mc-hardness}
$\MC(\tau\text{-}\FONK{\omega}{\omega})$ is $\AEXPPOLY$-hard for all vocabularies $\tau$, even on sentences and for a fixed $\tau$-structure $\calA$ with domain $\{0,1\}$ and a fixed team $\{\emptyset\}$.
\end{lemma}
\begin{proof}
We reduce from $\SAT(\PTL)$ (see Theorem~\ref{thm:sat-ptl}).
Given $\varphi \in \PTL$, suppose $\Prop(\varphi) = \{p_1,\ldots,p_n\}$.
The idea is that a team of worlds, or rather its induced set of Boolean assignments to $p_1,\ldots,p_n$, is simulated by a corresponding team of first-order assignments $s \colon X \to B$, where $X = \{z,x_1,\ldots,x_n\}$ and $B \dfn \{0,1\}$.
Here, the auxiliary variable $z$ plays the role of the constant $1$.
Let $V_b = (\{\emptyset\}^{z}_{\{b\}})^{x_1}_B\cdots^{x_n}_B$ for $b \in B$, \ie, the team $\{\emptyset\}$ after $n+1$ supplementations with $\{b\}$ resp.\ $B$.

\smallskip

By definition of vocabulary, either ${=} \in \tau$ or $\tau$ contains a predicate.
First, we consider the case ${=} \in \tau$.
We map $\varphi$ to $(\calA,\{\emptyset\},\psi)$, where $\dom \calA = B$,
$\psi \dfn \exists z\, \forall x_1 \cdots \forall x_n\, \top \lor \varphi^*$, and $\varphi^*$ is obtained from $\varphi$ by replacing each $p_i$ by $x_i = z$.
We prove that the reduction is correct.

To begin with, observe that
\begin{align*}
(\exists U \subseteq V_1 : (\calA,U) \vDash \varphi^*) \Leftrightarrow (\calA,V_1) \vDash \top \lor \varphi^* \Leftrightarrow (\calA,\{\emptyset\})\vDash \psi\text{.}\tag{$\star$}
\end{align*}
Here, "$\Rightarrow$" is clear.
We prove "$\Leftarrow$":
If $(\calA,\{\emptyset\})\vDash \psi$, then $(\calA,U) \vDash \varphi^*$ for some $U \subseteq V_0 \cup V_1$, but for each $s \in U \cap V_0$ we can simply flip all values of $s$.
It is easily proven with Lemma~\ref{lem:expansion} that this does not change the truth of the quantifier-free formula $\varphi^*$.

\smallskip

Next, assume that $\varphi$ is satisfiable, \ie, $(\calK,T) \vDash \varphi$ for some Kripke structure with team $(\calK,T)$.
For each world $w\in T$, define $s_w \colon X \to B$ by $s_w(z) = 1$ and $s_w(x_i) = 1 \Leftrightarrow (\calK,w) \vDash p_i$.
Then $(\calK,w) \vDash p_i \Leftrightarrow (\calA,s_w) \vDash x_i = z$, and by induction on the syntax of $\varphi$ we obtain $(\calA,U) \vDash \varphi^*$, where $U \dfn \{ s_w  \mid w \in T \}$.
Since $U \subseteq V_1$, $(\star)$ yields $(\calA,\{\emptyset\}) \vDash \psi$.
The other direction is similar.

\smallskip

Next, consider the case where ${=} \notin \tau$; then $\tau$ contains a predicate $P$.
We define $\calA$ as above, but let $P^\calA \dfn \{ (1,\ldots,1) \}$.
Furthermore, $\psi \dfn \forall x_1 \cdots \forall x_n\, \top \lor \varphi^*$, and $\varphi^*$ is now $\varphi$ with $p_i$ replaced by $P(x_i,\ldots,x_i)$.
The remaining proof is similar to the previous one.
\end{proof}

Clearly, the standard translation of satisfiable formulas is itself satisfiable.
A converse result holds as well.
Loosely speaking, from a first-order structure (with team) satisfying $\stx(\varphi)$ we can reconstruct a Kripke model (with team) for $\varphi$.

\begin{lemma}\label{lem:reverse-trans}
If $\varphi \in \MTL$, then $\varphi$ is satisfiable if and only if $\stx(\varphi)$ is satisfiable.
\end{lemma}
\begin{proof}
As Theorem~\ref{thm:translation} proves "$\Rightarrow$", we consider "$\Leftarrow$".
Suppose $(\calB,S) \vDash \stx(\varphi)$.
Then $\calB$ interprets the binary predicate $R$ and unary predicates $P_1,P_2,\ldots$.
By Proposition~\ref{prop:free-vars}, \wloss $S$ has domain $\{x\}$, \ie, $S = (x \langle S \rangle)^x$.

Define now the Kripke structure $\calK = (\dom \calB, R^\calB, V)$ such that $V(p_i) \dfn P_i^\calB$.
Then clearly $\calA(\calK) = \calB$.
By Theorem~\ref{thm:translation}, $(\calK,x\langle S \rangle) \vDash \varphi$.
\end{proof}

\subsection*{Completeness results}

\begin{theorem}\label{thm:mc-completeness}
Let $\calD$ be a p-uniform set of dependencies and $\tau$ a vocabulary.
\begin{itemize}
\item $\MC(\tau\text{-}\FONKD{\omega}{\omega})$ is $\AEXPPOLY$-complete, with hardness also on sentences and for a fixed $\tau$-structure $\calA$ with domain $\{0,1\}$ and a fixed team $\{\emptyset\}$.
\item If $\tau$ contains infinitely many relations and at least one of $k \geq 0,n\geq 1$ is finite, then $\MC(\tau\text{-}\FONKD{n}{k})$ is $\PSPACE$-complete.
\end{itemize}
\end{theorem}
\begin{proof}
The upper bounds are due to Corollary~\ref{cor:fo-mc} and~\ref{cor:mcfo2-upper1}, since alternating polynomial time coincides with $\PSPACE$ \cite{alternation}.
The lower bounds are due to Lemma~\ref{lem:fo-mc-hardness-nk} and~\ref{lem:fo-mc-hardness}.
\end{proof}

In particular, the model checking problem for first-order team logic, independence logic, inclusion and exclusion logic with Boolean negation is $\AEXPPOLY$-complete.
Likewise, their two-variable fragments are $\PSPACE$-complete.
This implies a straightforward generalization of a result of \citet[Proposition 33]{Lohrey12}:

\begin{corollary}
$\MC(\tau\text{-}\SO)$ is $\AEXPPOLY$-complete for all vocabularies $\tau$, with hardness also on sentences and with a fixed $\tau$-interpretation $\calA$ with $\dom \calA = \{0,1\}$.
\end{corollary}
\begin{proof}
By Proposition~\ref{prop:so-mc}, and reduction from $\MC(\tau\text{-}\FO(\negg))$.
Let $R$ be a $0$-ary relation variable.
In the spirit of Corollary~\ref{cor:fo-mc}, we map $(\calA,\{\emptyset\},\varphi)$ to $(\calA,\emptyset,\exists R \; \eta^{\emptyset}_\varphi(R) \land R)$.
\end{proof}

The next theorem states the main result of this paper.

\begin{theorem}\label{thm:main-sat}
Let $\tau$ contain at least one binary predicate, infinitely many unary predicates, and no functions.
\begin{itemize}
\item $\SAT(\tau\text{-}\FONN)$ and $\VAL(\tau\text{-}\FONN)$ are $\TOWERPOLY$-complete.
\item $\SAT(\tau\text{-}\FONK{2}{k})$ and $\VAL(\tau\text{-}\FONK{2}{k})$ are $\KAXX{k+1}$-hard if $k \geq 0$.
\end{itemize}
\end{theorem}
\begin{proof}
The upper bound for $\tau\text{-}\FONN$ is by Corollary~\ref{cor:fo2-upperbound}.
For the lower bounds, the mapping $\varphi \mapsto \stx(\varphi)$ is a reduction from $\SAT(\MTL)$ resp.\ $\SAT(\MTL_k)$ (see Theorem~\ref{thm:mtl-sat} and Lemma~\ref{lem:reverse-trans}).
The validity case is similar.
\end{proof}

\section{Conclusion}

In this paper, we proved that the logic $\FONN$ is complete for the class $\TOWERPOLY$ and hence decidable (Theorem~\ref{thm:main-sat}).
In particular, it has the finite model property (Theorem~\ref{thm:fo2-modelsize}), but exhibits non-elementary succinctness compared to classical $\FO^2$, which enjoys an exponential model property~\cite{fo2}.

Since validity is undecidable for two-variable dependence logic $\D^2$ \cite{il2} but not for $\FONN$, we conclude that team semantics with the dependence atom $\depop$ leads to undecidability, but with the Boolean negation $\negg$ it leads only to non-elementary blowup.

\smallskip

For $\FONKD{n}{k}$, where $n \geq 1$ and $k \geq 0$, we proved a dichotomy regarding its model checking complexity (Theorem~\ref{thm:mc-completeness}):
It is $\AEXPPOLY$-complete if $n = k = \omega$, and otherwise $\PSPACE$-complete.
However, both results hold for arbitrary sets of dependence atoms that are (polynomial time uniformly) definable in $\FO$.

In particular, $\MC(\TL)$ is $\AEXPPOLY$-complete, which naturally extends the result of \citet{mc-games} that $\MC(\D)$ is $\NEXPTIME$-complete.
However, the former model checking problem stays $\AEXPPOLY$-complete even without the dependence atom, while ordinary $\MC(\FO)$ becomes $\PSPACE$-complete for both classical and team semantics \cite{mc-games}.

\smallskip

In order to prove some of the lower bounds, we introduced a "standard translation" from $\MTL$ to $\FONN$ that extends the usual translation from $\ML$ to $\FO^2$, and also established $\KAXX{k+1}$-hardness of the satisfiability problem of $\FONK{2}{k}$.
Here, the upper bound is still missing.
In the modal setting, every satisfiable $\MTL_k$-formula has a $(k+1)$-fold exponential model.
It would be interesting to learn whether $\FONK{2}{k}$ has a similar small model property.
Together with Corollary~\ref{cor:mcfo2-upper1}, a positive answer would immediately yield a matching $\KAXX{k+1}$ upper bound.
 
\printbibliography

\clearpage

\appendix

\markboth{Appendix}{M. Lück}

\section{Proof details}

\begin{repproposition}{prop:rel}
Let $\calA$ be a structure, $\vv{t}$ a tuple of terms, $X \supseteq \Fr(\vv{t}\,)$.
For $i \in \{1,2\}$, let $T_i$ be a team in $\calA$ with domain $X_i \supseteq X$.
Then $\rest{T_1}{X} = \rest{T_2}{X}$ implies $\vv{t}\langle T_1 \rangle = \vv{t}\langle T_2 \rangle$.
Furthermore, for any tuple $\vv{x} \subseteq X$ of variables, $\vval{x}{T_1} = \vval{x}{T_2}$ iff $\rest{T_1}{\vv{x}} = \rest{T_2}{\vv{x}}$.
\end{repproposition}
\begin{proof}
For the first part of the proposition, we assume $\rest{T_1}{X} = \rest{T_2}{X}$, and (exploiting symmetry) show that $\vec{t}\langle T_1 \rangle \subseteq \vec{t}\langle T_2 \rangle$.
Let $\vec{a}  \in \vv{t}\langle T_1 \rangle$ be arbitrary.
Then $\vv{a} = \vv{t}\langle s\rangle$ for some $s \in T_1$.
By assumption, there is $s' \in T_2$ such that $\rest{s}{X} = \rest{s'}{X}$.
Since $\Fr(\vv{t}) \subseteq X$, clearly $\vv{t}\langle s \rangle = \vv{t}\langle s' \rangle$.
Consequently, $\vv{a} \in \vv{t}\langle T_2 \rangle$.

For the second part, suppose $\vval{x}{T_1} = \vval{x}{T_2}$ and let $s \in \rest{T_1}{\vv{x}}$ be arbitrary.
We show $s \in \rest{T_2}{\vv{x}}$, which suffices due to symmetry.
First, $s = \rest{s'}{\vv{x}}$ for some $s' \in T_1$.
Then $\vval{x}{s} = \vval{x}{s'} \in \vval{x}{T_1} = \vval{x}{T_2}$, and consequently, $\vval{x}{s} \in \vval{x}{T_2}$.
But then $\vval{x}{s} = \vval{x}{s''}$ for some $s'' \in T_2$, which implies $s = \rest{s''}{\vv{x}}$, and accordingly $s \in \rest{T_2}{\vv{x}}$.
\end{proof}

\begin{repproposition}{prop:iso}
Let $\calA$ be a structure, $\vv{x}$ a tuple of variables, and $V \dfn \{ s \colon \vv{x} \to \calA \}$.
Then $\pow{V}$ is the set of all teams in $\calA$ with domain $\vv{x}$, and the mapping $S \mapsto \vv{x}\langle S\rangle$ is an isomorphism
between the Boolean lattices $(\pow{V},\subseteq)$ and $(\pow{\calA^{\size{\vv{x}}}},\subseteq)$.
\end{repproposition}
\begin{proof}
Let $n \dfn \size{\vv{x}}$.
Clearly, every team with domain $\vv{x}$ is in $\pow{V}$, with $\pow{V}$ being a Boolean lattice \wrt $\subseteq$.
It is easy to show that $r$ is surjective:
Given $A \subseteq (\dom \calA)^n$, define the team $S \dfn \{ s \in V \mid \vv{x}\langle s\rangle \in A\}$.
Then $r(S) = \{ \vv{x}\langle s\rangle \mid s \in V \text{ and } \vv{x}\langle s\rangle \in A\} = A$.

Moreover, $r$ preserves $\subseteq$ in both directions:
Suppose $S \subseteq S'$ and let $\vv{a} = (a_1,\ldots,a_n) \in r(S)$ be arbitrary.
We show $\vv{a} \in r(S')$, as then $r(S) \subseteq r(S')$.
Since $\vv{a} \in r(S) = \vv{x}\langle S\rangle$, there exists $s \in S$ such that $\vv{x}\langle s\rangle = \vv{a}$.
As $s \in S'$, consequently, $\vv{a} \in \vv{x}\langle S'\rangle = r(S')$.

Conversely, suppose $r(S) \subseteq r(S')$ and let $s \in S$ be arbitrary.
As $\vv{x}\langle s\rangle \in r(S) \subseteq r(S') = \vv{x}\langle S'\rangle$, there exists an assignment $s' \in S'$ such that $\vv{x}\langle s\rangle = \vv{x}\langle s'\rangle$.
However, as $\dom s = \dom s' = \vv{x}$, necessarily $s = s'$, \ie, $s \in S'$.
As $S \subseteq S' \Leftrightarrow r(S) \subseteq r(S')$, and $r$ is surjective, we conclude that $r$ is also injective and hence an isomorphism.
\end{proof}

\begin{repproposition}{prop:def-of-supplement}
Let $T$ be a team with domain $X$ and $S$ a team with domain $X \cup \{x\}$ (with possibly $x \in X$), and let $X' \dfn X \setminus \{x\}$.
Then $\rest{S}{X'} = \rest{T}{X'}$ if and only if there is a supplementing function $f$ such that $S = T^x_f$.
\end{repproposition}
\begin{proof}
Let $\calA$ be the underlying structure.

"$\Rightarrow$":
Suppose $\rest{S}{X'} = \rest{T}{X'}$.
First, we show that for every $s' \in S$ there is $s \in T$ such that $s' = s^x_a$ for some $a$.
By assumption, $\rest{s'}{X'} = \rest{s}{X'}$ for some $s \in T$.
But then $s' = s^x_{s'(x)}$.

Next, we define the function $f(s) \dfn \Set{ a \in \calA | s^x_a \in S }$.
Then $T^x_f = \{ s^x_a \mid s \in T, a \in f(s) \} = \{ s^x_a \mid s \in T, s^x_a \in S \}$, which equals $S$ by the above argument.
In order to show that $f$ is a supplementing function, it remains to show that $f(s) \neq \emptyset$ for all $s \in T$, \ie, that for every $s\in T$ there exists $a \in \calA$ such that $s^x_a \in S$.
This follows again by $\rest{S}{X'} = \rest{T}{X'}$.

\smallskip

"$\Leftarrow$":
First, we show "$\subseteq$", \ie, that $s \in \rest{T}{X'}$ for arbitrary
$s \in \rest{S}{X'}$.
By definition, for such $s$ we have $s = \rest{s'}{X'}$ for some $s' \in S$.
Since $S = T^x_f$, there exists $s'' \in T$ and $a \in \calA$ such that $s' = (s'')^x_a$.
As $y \notin X'$, we have $s = \rest{s'}{X'} = \rest{s''}{X'} \in \rest{T}{X'}$.

For the other direction, "$\supseteq$", let $s \in \rest{T}{X'}$ be arbitrary.
Then $s = \rest{s'}{X'}$ for some $s' \in T$.
As $S = T^x_f$, there exists some $s'' \in S$ and $a \in \calA$ such that $s'' = (s')^x_a$.
Again we have $s = \rest{s'}{X'} = \rest{s''}{X'}$, \ie, $s \in \rest{S}{X'}$.
\end{proof}

\begin{algorithm}[t]
\DontPrintSemicolon
\SetKwInOut{Input}{Input}
\SetKwInOut{Output}{Output}

\nonl\textbf{Algorithm:} $\mathsf{check}(\alpha,\calA,\calJ)$ for $\alpha\in \tau\text{-}\SO$ in negation normal form, a $\tau$-structure $\calA$, and a second-order interpretation $\calJ$ of $\Fr(\alpha)$.

\vspace{5pt}

\If{$\alpha$ is an atomic formula or the negation of an atomic formula}{
\Return{\emph{\texttt{true}} if $(\calA,\calJ) \vDash \alpha$ and \emph{\texttt{false}} otherwise};
}
\lElseIf{$\alpha = \gamma_1 \lor \gamma_2$}{
existentially choose $i \in \{1,2\}$ and let $\alpha \leftarrow \gamma_i$
}
\lElseIf{$\alpha = \gamma_1 \land \gamma_2$}{
universally choose $i \in \{1,2\}$ and let $\alpha \leftarrow \gamma_i$
}
\ElseIf{$\alpha = \Game X \gamma$ for $\Game \in \{\exists, \forall\}$ and $X \in \Var(\alpha)$}{
$\alpha \leftarrow \gamma$\\
\If{$X \in \Fr(\gamma)$}{
\leIf{$\Game = \exists$}{switch to existential branching}{switch to universal branching}
\If{$X$ is a first-order variable}{
non-deterministically choose $a \in \size{\calA}$ and let $\calJ \leftarrow \calJ^X_a$
}
\ElseIf{$X$ is a function variable}{
non-deterministically choose $F \colon \size{\calA}^{\arity{X}} \to \size{\calA}$ and let $\calJ \leftarrow \calJ^X_F$
}
\ElseIf{$X$ is a relation variable}{
non-deterministically choose $B \subseteq \size{\calA}^{\arity{X}}$ and let $\calJ \leftarrow \calJ^X_B$
}
}
}
\Return{$\mathsf{check}(\alpha,\calA,\rest{\calJ}{\Fr(\alpha)})$}
\medskip
\caption{Algorithm for $\MC(\SO)$\label{alg:mc-so}}
\end{algorithm}

\begin{repproposition}{prop:so-mc}
Model checking of $\SO$ is decidable in time $2^{n^\bigO{1}}$ and with at most $\qr(\alpha)$ alternations, where $(\calA,\calJ,\alpha)$ is the input.
\end{repproposition}
\begin{proof}
\Wloss $\alpha$ is in negation normal form, \ie, $\neg$ appears only in front of atomic formulas.
Let $A \dfn \dom \calA$.
We abbreviate
\[
\size{\calJ} \dfn \sum_{\substack{X \in \dom \calJ\\X \text{ second-order}}} \size{\calJ(X)}\text{,}
\]
\ie, the sum of the sizes of functions and relations in $\calJ$.

We replace $\calJ$ by $\calJ' \dfn \rest{\calJ}{\Fr(\alpha)}$, which can be done in time polynomial in $\size{\calJ}$ and $\size{\alpha}$.

Since for second-order $X$ we then have $\size{\calJ'(X)} \leq \size{A}^{\arity{X}} \leq \size{A}^{\size{\alpha}}$, and also $\size{\dom \calJ'} = \size{\Fr(\alpha)} \leq \size{\alpha}$, the sum $\size{\calJ'}$ is then bounded from above by $\size{\alpha} \cdot \size{A}^{\size{\alpha}}$.

Now we run Algorithm~\ref{alg:mc-so}, which performs at most $\size{\alpha}$ recursive calls, and clearly at most $\size{\alpha}$ alternations.
Furthermore, the $i$-th recursive call is of the form $\mathsf{check}(\alpha_i,\calA,\calJ_i)$ with $\size{\alpha_i} \leq \size{\alpha}$ and, by the same argument as before, $\size{\calJ_i} \leq \size{\alpha} \cdot \size{A}^\size{\alpha}$.

For this reason, it is easy to see that the overall runtime is polynomial in $\size{\calJ}$ and $\size{A}^{\size{\alpha}}$.
\end{proof}

\begin{replemma}{lem:supp-as-rest}
Let $T$ have domain $\vv{x}$ and $S$ have domain $\vv{x} \cup \{y\}$ (with possibly $y \in X$), and let $X' \dfn \vv{x} \setminus \{x\}$.
Then $\rest{T}{X'} = \rest{S}{X'}$ if and only if $\calA \vDash \pi(\vv{x}\langle T\rangle,\vv{x};y\langle S\rangle)$, where $\pi(T,S) \dfn  \forall \vv{x}((\exists y T \vv{x}) \leftrightarrow (\exists y S \vv{x};y))$.
\end{replemma}
\begin{proof}
First, let us consider the case $y \notin X$, \ie, $X' = X$.
Then:
\begin{align*}
&\rest{T}{X'} = \rest{S}{X'}\\
\Leftrightarrow\;& \vval{x}{T} = \vval{x}{S}\tag{by Proposition~\ref{prop:rel}}\\
\Leftrightarrow\;& \forall \vv{a} : (\vv{a} \in \vval{x}{T} \Leftrightarrow \exists b : (\vv{a},b)\in (\vv{x};y)\langle S \rangle)\tag{since $T$ has domain $\vv{x}$}\\
\Leftrightarrow\;& \calA \vDash \psi(\vv{x}\langle T\rangle,(\vv{x};y)\langle S\rangle)\tag{since $\exists y T\vv{x} \equiv T\vv{x}$}
\end{align*}
Next, assume $y \in X$ and \wloss $y = x_n$.
Then $\vv{x};y = \vec{x}$ and $X' = \{x_1,\ldots,x_{n-1}\}$.
Let $\vv{x}' = (x_1,\ldots,x_{n-1})$.
Analogously as before, we have:
\begin{align*}
&\rest{T}{X'} = \rest{S}{X'}\\
\Leftrightarrow\;& \vv{x}'\langle T\rangle = \vv{x}'\langle S\rangle\\
\Leftrightarrow\;& \forall \vv{a} : \big((\exists b : (\vv{a},b) \in \vval{x}{T}) \Leftrightarrow (\exists b : (\vv{a},b)\in (\vv{x};y)\langle S \rangle)\big)\\
\Leftrightarrow\;& \calA \vDash \psi(\vv{x}\langle T\rangle,(\vv{x};y)\langle S\rangle)\qedhere
\end{align*}
\end{proof}

\end{document}